\documentclass[11pt]{article}
\usepackage[letterpaper,margin=1in]{geometry}
\usepackage{xcolor}
\usepackage[colorlinks]{hyperref}
\definecolor{lightblue}{rgb}{0.5,0.5,1.0}
\definecolor{darkred}{rgb}{0.5,0,0}
\definecolor{darkgreen}{rgb}{0,0.5,0}
\definecolor{darkblue}{rgb}{0,0,0.5}

\hypersetup{colorlinks,linkcolor=darkred,filecolor=darkgreen,urlcolor=darkred,citecolor=darkblue}

\usepackage{tabularx}
\usepackage{mathtools}
\usepackage{amsthm}
\usepackage{thmtools}
\usepackage{verbatim}
\usepackage{enumitem}
\usepackage{authblk}
\usepackage{amssymb}
\usepackage{amsmath}
\usepackage{caption}
\usepackage{subcaption}
\usepackage{nicefrac}
\usepackage{mathtools}
\usepackage[linesnumbered,ruled,commentsnumbered,longend]{algorithm2e}
\usepackage{tikz}
\usetikzlibrary{graphs}
\usetikzlibrary{shapes.misc, positioning}
\usetikzlibrary{shapes,fit}
\usepackage[edges]{forest}
\usetikzlibrary{shapes,arrows,arrows.meta,calc,decorations.pathmorphing}
\usepackage{amsmath}
\usepackage{todonotes} 

\newcommand{\Traces}{\textsc{Traces}}

\newcommand{\saucy}{\textsc{saucy}}
\newcommand{\dejavu}{\textsc{dejavu}}

\newcommand{\sassy}{\textsc{sassy}}

\newtheorem{fact}{Fact}

\DeclareMathOperator{\Var}{Var}

\DeclareMathOperator{\Lit}{Lit}

\makeatletter
\newcommand{\Letter}[1]{\@alph{#1}}
\makeatother

\SetKwProg{Fn}{function}{}{end}\SetKwFunction{RandomAut}{Automorphisms}
\SetKwProg{Fn}{function}{}{end}\SetKwFunction{RandomIso}{Isomorphism}
\SetKwFunction{FirstLeaf}{FirstLeaf}
\SetKwFunction{RandomWalk}{RandomWalk}
\SetKwFunction{NextNodeDFS}{OneNodeDFS}
\SetKwFunction{NextNodeBFS}{OneNodeBFS}
\SetKwFunction{Some}{Some}
\SetKwFunction{Sift}{Sift}
\SetKwFunction{Overlap}{CycleOverlap}
\SetKwFunction{OrbitGraph}{OrbitGraph}
\SetKwFunction{SymmetricAction}{SymmetricAction}
\SetKwFunction{Equivalent}{EquivalentOrbits}
\SetKwFunction{RandomElement}{RandomElement}
\SetKwFunction{RandomChild}{RandomChild}
\SetKwFunction{NextChild}{NewChild}
\SetKwFunction{NextRandomChild}{NewRandomChild}
\SetKwFunction{Refine}{Refine}
\SetKwFunction{Subtree}{Subtree}
\SetKwFunction{RRef}{Ref}
\SetKwFunction{SSel}{Sel}
\SetKwFunction{Sift}{Sift}
\SetKwFunction{CertifyAutomorphism}{CertifyAutomorphism}
\SetKwFunction{CertifyIsomorphism}{CertifyIsomorphism}

\newcommand{\n}[1]{\overline{#1}}



\DeclareMathOperator{\Sym}{Sym}
\DeclareMathOperator{\Alt}{Alt}
\DeclareMathOperator{\Aut}{Aut}

\DeclareMathOperator{\supp}{supp}
\DeclareMathOperator{\enc}{enc}

\newtheorem{lemma}{Lemma}
\newtheorem{corollary}[lemma]{Corollary}
\newtheorem{theorem}[lemma]{Theorem}

\newtheorem{definition}[lemma]{Definition}

\title{Algorithms Transcending the SAT-Symmetry Interface}

\author[1]{Markus Anders}
\author[1]{Pascal Schweitzer}
\author[2]{Mate Soos}
\affil[1]{TU Darmstadt}
\affil[2]{National University of Singapore}

\newcommand\blfootnote[1]{%
  \begingroup
  \renewcommand\thefootnote{}\footnote{#1}%
  \addtocounter{footnote}{-1}%
  \endgroup
}

\begin{document}

\maketitle

\begin{abstract}
Dedicated treatment of symmetries in satisfiability problems (SAT) is indispensable for solving various classes of instances arising in practice. However, the exploitation of symmetries usually takes a black box approach. Typically, off-the-shelf external, general-purpose symmetry detection tools are invoked to compute symmetry groups of a formula. The groups thus generated are a set of permutations passed to a separate tool to perform further analyzes to understand the structure of the groups. The result of this second computation is in turn used for tasks such as static symmetry breaking or dynamic pruning of the search space. Within this pipeline of tools, the detection and analysis of symmetries typically incurs the majority of the time overhead for symmetry exploitation.

In this paper we advocate for a more holistic view of what we call the   \emph{SAT-symmetry interface}. We formulate a computational setting, centered around a new concept of joint graph/group pairs, to analyze and improve the detection and analysis of symmetries. Using our methods, no information is lost performing computational tasks lying on the SAT-symmetry interface.
Having access to the entire input allows for simpler, yet efficient algorithms.

Specifically, we devise algorithms and heuristics for computing finest direct disjoint decompositions, finding equivalent orbits, and finding natural symmetric group actions. Our algorithms run in what we call instance-quasi-linear time, i.e., almost linear time in terms of the input size of the original formula and the description length of the symmetry group returned by symmetry detection tools.  Our algorithms improve over both heuristics used in state-of-the-art symmetry exploitation tools, as well as theoretical general-purpose algorithms.
\end{abstract}
\blfootnote{Supported by the European Research Council (ERC) under the European Union's Horizon 2020 research and innovation programme (EngageS: grant No.~{820148}).}

\section{Introduction}
Many SAT instances, especially of the hard combinatorial type, exhibit symmetries. When symmetries exhibited by these instances
are not handled adequately, SAT solvers may repeatedly explore symmetric
parts of the search space. 
This can dramatically increase runtime, sometimes making it impossible for the solver to finish within reasonable time \cite{DBLP:journals/siamcomp/BeameKPS02}.

One common method to handle the symmetries is to add symmetry breaking
formulas to the problem specification~\cite{DBLP:conf/kr/CrawfordGLR96, DBLP:conf/dac/AloulRMS02}. This approach is called static symmetry breaking.
Another, competing, approach is to handle symmetries dynamically during the running of the SAT solver. There are a variety of such dynamic strategies, exploiting symmetry information during variable branching~\cite{DBLP:conf/cp/KirchwegerS21} and learning~\cite{DBLP:journals/constraints/Sabharwal09, DBLP:conf/sat/Devriendt0B17}.
For SAT, the tools \textsc{Shatter}~\cite{DBLP:conf/dac/AloulMS03} and \textsc{BreakID}~\cite{DBLP:conf/sat/Devriendt0BD16, DBLP:conf/aaai/0001GMN22} take the static symmetry breaking approach, while \textsc{SymChaff}~\cite{DBLP:journals/constraints/Sabharwal09} and \textsc{SMS}~\cite{DBLP:conf/cp/KirchwegerS21} take the dynamic symmetry exploitation approach. 

\bigskip
While there is a growing number of competing approaches of how best to handle symmetries, there are also a number of common obstacles:
symmetries of the underlying formula have to be detected first, and the structure of symmetries has to be understood, at least to some degree.
Approaches that handle symmetries can be typically divided into three distinct steps: (Step~1) symmetry detection, (Step~2) symmetry analysis, and (Step~3) symmetry breaking, or other ways of exploiting symmetry. 
In the following, we discuss each of these steps, also illustrated on the left side of Figure~\ref{fig:satsyminterface}.

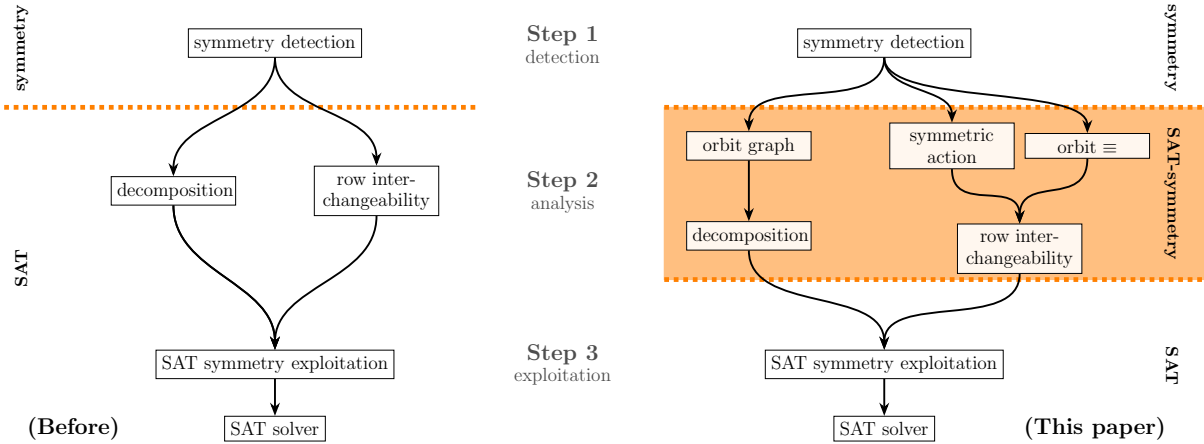
\begin{figure}
  \centering
  \scalebox{0.6}{
  \begin{tikzpicture}
  \begin{scope}[scale=3,yscale=0.95]
    \node[anchor=west] at (2-2+0.125, -0.5)  (time)  {\color{black} \Large \textbf{(Before)}};

    \node[draw, color=black,rectangle] at (2, 2+0.5)  (sd)  {\large symmetry detection};
    \node[draw, color=black,rectangle,text width=2.5cm, align=center] at (2-0.75, 1.35)  (gd)  {\large decomposition};
    \node[draw, color=black,rectangle,text width=2.5cm, align=center] at (2+0.75, 1.35)  (ci)  {\large row interchangeability};
    \node[draw, color=black,rectangle] at (2, 0)  (se)  {\large SAT symmetry exploitation};
    \node[draw, color=black,rectangle] at (2, -0.5)  (sas)  {\large SAT solver};
    
    \draw[line width=3pt,draw=orange,dashed] (2-2, 2) to (2-2+2+2-0.5, 2);
    
    \node[color=black,rotate=90] at (2-2+0.125, 0.75)  (sasymi)  {\textbf{SAT}};
    \node[color=black,rotate=90] at (2-2+0.125, 0.125+1+1.34)  (sasymi)  {\textbf{symmetry}};
  
    \draw[-{Stealth[scale=1]}, very thick] (sd) to[out=270,in=90] (gd);
    \draw[-{Stealth[scale=1]}, very thick] (sd) to[out=270,in=90]  (ci);
  
    \draw[-{Stealth[scale=1]}, very thick] (gd) to[out=270,in=90]  (se);
    \draw[-{Stealth[scale=1]}, very thick] (ci) to[out=270,in=90]  (se);
  
    \draw[-{Stealth[scale=1]}, very thick] (gd) to[out=270,in=90]  (se);
    \draw[-{Stealth[scale=1]}, very thick] (se) to[out=270,in=90]  (sas);
  \end{scope}
  \begin{scope}[scale=3,yscale=0.95,xshift=4.125cm]
  \node[color=black!70,rotate=0] at (0, 0.125+1+1.34)  (s1)  {\begin{tabular}{c}
  \Large\textbf{Step~1}\\
  \large detection
  \end{tabular}};
  \node[color=black!70,rotate=0] at (0, 0.33+1)        (s2)  {\begin{tabular}{c}
    \Large\textbf{Step~2}\\
    \large analysis
    \end{tabular}};
  \node[color=black!70,rotate=0] at (0, 0.33+1-1.34)   (s3)  {\begin{tabular}{c}
    \Large\textbf{Step~3}\\
    \large  exploitation
    \end{tabular}};
\end{scope}
  \begin{scope}[scale=3,yscale=0.95,xshift=4.5cm]
    \node[anchor=east] at (2+2+0.125, -0.5)  (time)  {\Large \color{black} \textbf{(This paper)}};

    \draw[draw=orange!50,fill=orange!50] (2-2+0.375,0.66) rectangle ++(2+2,2-0.66);
    \node[rotate=270] at (2+2+0.125, 0.33+1)  (sasymi)  {\textbf{SAT-symmetry}};

    \draw[line width=3pt,draw=orange,dashed] (2-2+0.375, 0.66) to (2-2+0.375+2+2, 0.66);
    \draw[line width=3pt,draw=orange,dashed] (2-2+0.375, 2) to (2-2+0.375+2+2, 2);

    \node[color=black,rotate=270] at (2+2+0.125, 0.33+1-1.34)  (sasymi)  {\textbf{SAT}};

    \node[color=black,rotate=270] at (2+2+0.125, 0.125+1+1.34)  (sasymi)  {\textbf{symmetry}};

    \node[draw,color=black,rectangle] at (2, 2+0.5)  (sd)  {\large symmetry detection};
    \node[draw, fill=white!75!orange!25,rectangle,text width=2.5cm, align=center] at (2-0.5+1, 1.7)  (gd)  {\large symmetric action};
    \node[draw, fill=white!75!orange!25,rectangle,text width=2.5cm, align=center] at (2+0.5+1, 1.7)  (ci)  {\large orbit $\equiv$};
    \node[draw, fill=white!75!orange!25,rectangle,text width=2.5cm, align=center] at (2-1, 1.7)  (orbitg)  {\large orbit graph};
    \node[draw, fill=white!75!orange!25,rectangle,text width=2.5cm, align=center] at (2-1, 1)  (ga)  {\large decomposition};

    \node[draw, fill=white!75!orange!25,rectangle,text width=2.5cm, align=center] at (2+1, 0.9)  (cix)  {\large row interchangeability};

    \node[draw=cyan!50, color=black,rectangle] at (2, 0)  (se)  {\large SAT symmetry exploitation};
    \node[draw=black, color=black,rectangle] at (2, -0.5)  (sas)  {\large SAT solver};
    
    \node[] at (2-2, 0.25+0.375-0.35)  (li)  {};
    \node[] at (2+2, 0.25+0.375-0.35)  (ri)  {};
  
  
    \draw[-{Stealth[scale=1]}, very thick] (sd) to[out=270,in=90,looseness=1] (gd);
    \draw[-{Stealth[scale=1]}, very thick] (sd) to[out=270,in=90,looseness=0.6] (ci);
    \draw[-{Stealth[scale=1]}, very thick] (sd) to[out=270,in=90,looseness=1] (orbitg);

    \draw[-{Stealth[scale=1]}, very thick] (orbitg) to[out=270,in=90,looseness=0.7] (ga);
  
    \draw[-{Stealth[scale=1]}, very thick] (gd) to[out=270,in=90,looseness=1.4] (cix);
    \draw[-{Stealth[scale=1]}, very thick] (cix) to[out=270,in=90] (se);
    \draw[-{Stealth[scale=1]}, very thick] (ci) to[out=270,in=90,looseness=1.4] (cix);
    \draw[-{Stealth[scale=1]}, very thick] (ga) to[out=270,in=90] (se);
  
    \draw[-{Stealth[scale=1]}, very thick] (se) to[out=270,in=90] (sas);
  \end{scope}
\end{tikzpicture}
  }
  \caption{Blurring the lines of the SAT-symmetry interface. We analyze existing practical routines (left), draw connections to existing concepts in computational group theory, and describe improved algorithms in our new SAT-symmetry context (right).} \label{fig:satsyminterface}
\end{figure}

\emph{Step~1.} In practice, symmetries are detected by modeling a given SAT formula as a graph, and then applying an off-the-shelf symmetry detection tool, such as \saucy{}~\cite{DBLP:conf/dac/DargaLSM04}, 
to the resulting graph.
Since symmetries form a permutation group under composition, a symmetry detection tool does not return all the symmetries. Instead, it only returns a small set of \emph{generators}, which, by composition, lead to all the symmetries of the formula.
Indeed, returning only a small set of generators is crucial for efficiency, since the number of symmetries is often exponential in the size of the formula.

\emph{Step~2.} Symmetry exploitation algorithms apply heuristics to analyze the structure of the group described by the generators.
This is necessary to enable the best possible use of the symmetries to improve SAT solver performance.
We mention three examples for structural analyzes.
Firstly, the \emph{disjoint direct decomposition} splits a group into independent parts that can be handled separately.
Secondly, so-called \emph{row interchangeability} subgroups of the group~\cite{DBLP:conf/sat/Devriendt0BD16, DBLP:conf/cp/FlenerFHKMPW02, DBLP:journals/mpc/PfetschR19} are of particular interest since they form a class of groups for which linear-sized, complete symmetry breaking constraints are known. 
Thirdly, \emph{stabilizers} are commonly used for various purposes among both static and dynamic approaches~\cite{DBLP:conf/cp/Puget03}.

\emph{Step~3.} Lastly, the symmetries and structural insights are used to reduce the search space in SAT using one of the various static and dynamic symmetry exploitation approaches.

\bigskip
Designing symmetry exploitation algorithms typically involves delicately balancing computational overhead versus how thoroughly symmetries are used. 
In this trade-off, symmetry detection (Step~1) and analysis (Step~2) typically induce the majority of the overhead \cite{DBLP:conf/sat/Devriendt0BD16}. 
The main focus of this paper is improving the analysis of symmetries, i.e. (Step~2).

Practical implementations in use today that perform such structural analyzes do so through heuristics. While using heuristics is not an issue per se, some heuristics currently in use strongly depend on properties that the generators returned by symmetry detection tools may or may not exhibit.
For example, \textsc{BreakID} and the MIP heuristic in~\cite{DBLP:journals/mpc/PfetschR19} both rely on so-called \emph{separability} of the generating set and a specific arrangement of \emph{transpositions} being present. Neither of these properties are guaranteed by contemporary symmetry detection tools \cite{DBLP:journals/jsc/ChangJ22}.

In fact, modern symmetry detection tools such as \Traces{} \cite{DBLP:journals/jsc/McKayP14} and \dejavu{} \cite{DBLP:conf/esa/AndersS21} return randomly selected symmetries, since the use of randomization provides an asymptotic advantage in the symmetry detection process itself~\cite{DBLP:conf/icalp/AndersS21}. However, generating sets consisting of randomly selected symmetries are in a sense the exact opposite of what is desired for the heuristics, since with high probability random symmetries satisfy neither of the required conditions.
This is particularly unfortunate, as \dejavu{} is currently the fastest symmetry detection tool available for graphs stemming from SAT instances~\cite{DBLP:journals/corr/abs-2302-06351}.

Another downside of the use of practical heuristics for the structural analysis of the group is that they are often also computationally expensive
and make up a large portion of the runtime of the overall symmetry exploitation process.
For example, the row interchangeability algorithm of \textsc{BreakID} performs multiple callbacks to the underlying symmetry detection tool, where each call can be expensive.

Altogether, heuristics in use today sometimes cause significant overhead, while also posing an obstacle to speeding up symmetry detection itself.
This immediately poses the question: why is it that these heuristics are currently in place that cause such a loss of efficiency when it comes to computations within the SAT-symmetry interface?

We believe that the issue is that tools on either side of the interface treat each other as a black box. Indeed, when considered as an isolated task, algorithms for the analysis of permutation groups are well-researched in the area of computational group theory \cite{seress_2003}. Not only is the theory well-understood, but there are also highly efficient implementations \cite{GAP4}. However, we can make two crucial observations regarding the available algorithms. 
First and foremost, for group theoretic algorithms from the literature that are deemed to have linear or nearly-linear runtime~\cite{seress_2003}, the concrete runtime notions actually differ from the ones applicable in the overall context. In fact, the runtime is essentially measured in 
terms of a dense rather than a sparse input description. Therefore, in the context of SAT-solving or graph algorithms, the runtime of these algorithms should rather be considered quadratic. Secondly, in computational group theory, algorithms assume that only generators for an input group are available.
However, in the context of the SAT-symmetry interface, not only a group but also a graph (computed from the original formula) is available. It turns out as a key insight of our paper that lacking access to the graphs crucially limits the design space for efficient algorithms.

\noindent \textbf{Contributions.}
Advocating a holistic view of the SAT-symmetry interface, we develop algorithms that transcend both into  the SAT domain and the symmetry domain at the same time.
This is illustrated in Figure~\ref{fig:satsyminterface} on the right side. 

Firstly, we provide a definition for the computational setting such as input, output, and runtime, under which these algorithms should operate (Section~\ref{sec:setting}).
We then extract precise formal problem definitions from heuristics implemented in state-of-the-art tools (Section~\ref{sec:satandsym}).
Lastly, we demonstrate the efficacy of our new approach by providing faster theoretical algorithms for commonly used heuristics, as is described below. 

\noindent \textbf{Computational Setting.}
In our new computational setting, algorithms take as input a \emph{joint graph/group pair}, meaning a group $S$ and corresponding graph $G$, whose symmetry group is precisely $\langle S \rangle$.
We define a precise notion of \emph{instance-linear time}, meaning it is linear in the encoding size of the SAT formula, graph, and group.

\noindent \textbf{New Algorithms.} Given a joint graph/group pair, we develop and analyze the following algorithms:
\begin{enumerate}
  \item[\textbf{A1}] An instance-linear algorithm for computing the \emph{finest direct disjoint decomposition} of the symmetry group of a graph (Section~\ref{sec:decompose}). 
  We also give a heuristic specific to SAT formulas, decomposing the symmetry group on the literals.
  \item[\textbf{A2}] An algorithm to simultaneously detect \emph{natural symmetric group actions} on all the orbits of a group (Section~\ref{sec:snaction}). Here we exploit randomized techniques from computational group theory for the detection of ``giant'' permutation groups. 
  We give instance-linear heuristics which are able to exploit properties of the SAT-symmetry interface. 
  \item[\textbf{A3}] An instance-quasi-linear algorithm to compute \emph{equivalent symmetric orbits}, under some mild assumptions about the generating set (Section~\ref{sec:rowinterchange}).
  In conjunction with (A2), this enables us to detect \emph{all} elementary row interchangeability subgroups.
\end{enumerate}
Both (A1) and (A3) improve the (at least) quadratic runtime of previous, general-purpose permutation group algorithms of \cite{DBLP:journals/jsc/ChangJ22} and \cite{seress_2003}, respectively.

\section{Preliminaries and Related Work} \label{sec:preliminaries}

\noindent \textbf{Graphs and Symmetries.}
A colored graph $G = (V, E, \pi)$ consists of a set of vertices $V$, edges $E \subseteq V \times V$, and a vertex coloring $\pi \colon V \to C$ which maps $V$ to some set of colors $C$. We use $V(G)$, $E(G)$, and $\pi(G)$ to refer to the vertices, edges, and coloring of $G$, respectively.

A symmetry, or \emph{automorphism}, of a colored graph $G = (V, E, \pi)$ is a bijection $\varphi\colon V \to V$ such that $\varphi(E) = E$ as well as $\pi(v) = \pi(\varphi(v))$ for all $v \in V$.
In other words, symmetries preserve the neighborhood relation of the graph, as well as the coloring of vertices.
The colors of vertices in the graph are solely used to ensure that distinctly colored vertices are not mapped onto each other using symmetries. 
Together, all symmetries of a graph form a permutation group under composition, which we call the \emph{automorphism group} $\Aut(G)$. 

In this paper, we call software tools computing the automorphism group of a graph \emph{symmetry detection tools} \cite{DBLP:journals/jsc/McKayP14, DBLP:conf/dac/DargaLSM04, DBLP:conf/tapas/JunttilaK11, DBLP:journals/jsc/McKayP14, DBLP:conf/esa/AndersS21}. 
In the literature, these tools are also often called \emph{practical graph isomorphism solvers}. 
In this paper, we avoid the use of this term in order not to confuse them with SAT \emph{solvers}.

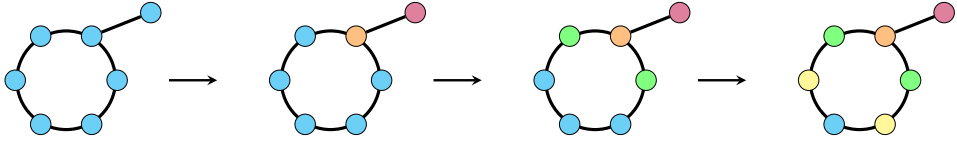
\begin{figure}
  \centering
\scalebox{0.75}{
\begin{tikzpicture}[scale=1.5]
 \node [circle, draw, minimum size=1.75cm, ultra thick] (c) at (0*1.25,0) {};
  \foreach \i in {0,...,5}{
    \ifthenelse{\i=1}{
      \node[draw, fill=cyan!50, circle] at (c.\i*360/6) (a\i)  {};
    }{
      \node[draw, fill=cyan!50, circle] at (c.\i*360/6) (a\i)  {};
    }
  }
  \node[draw, fill=cyan!50, circle] at (0.8*1.25,0.8) (a7)  {};
  \draw[ultra thick,] (a7) -- (a1);


  \node [circle, draw, minimum size=1.75cm, ultra thick] (c) at(2.5*1.25,0) {};
  \foreach \i in {0,...,5}{
    \ifthenelse{\i=1}{
      \node[draw, fill=orange!50, circle] at (c.\i*360/6) (a\i)  {};
    }{
      \node[draw, fill=cyan!50, circle] at (c.\i*360/6) (a\i)  {};
    }
  }
  \node[draw, fill=purple!50, circle] at (2.5*1.25+0.8*1.25,0.8) (a7)  {};
  \draw[ultra thick,] (a7) -- (a1);


  \node [circle, draw, minimum size=1.75cm, ultra thick] (c) at(5*1.25,0) {};
  \foreach \i in {0,...,5}{
    \ifthenelse{\i=1}{
      \node[draw, fill=orange!50, circle] at (c.\i*360/6) (a\i)  {};
    }{
      \ifthenelse{\i=0 \OR \i=2}{
      \node[draw, fill=green!50, circle] at (c.\i*360/6) (a\i)  {};
      }{
        \node[draw, fill=cyan!50, circle] at (c.\i*360/6) (a\i)  {};
      }
    }
  }
  \node[draw, fill=purple!50, circle] at (5*1.25+0.8*1.25,0.8) (a7)  {};
  \draw[ultra thick,] (a7) -- (a1);


  \node [circle, draw, minimum size=1.75cm, ultra thick] (c) at(7.5*1.25,0) {};
  \foreach \i in {0,...,5}{
    \ifthenelse{\i=1}{
      \node[draw, fill=orange!50, circle] at (c.\i*360/6) (a\i)  {};
    }{
      \ifthenelse{\i=0 \OR \i=2}{
      \node[draw, fill=green!50, circle] at (c.\i*360/6) (a\i)  {};
      }{
        \ifthenelse{\i=4}{
        \node[draw, fill=cyan!50, circle] at (c.\i*360/6) (a\i)  {};
        }{
          \node[draw, fill=yellow!50, circle] at (c.\i*360/6) (a\i)  {};
        }
      }
    }
  }
  \node[draw, fill=purple!50, circle] at (7.5*1.25+0.8*1.25,0.8) (a7)  {};
  \draw[ultra thick,] (a7) -- (a1);


   \foreach \i in {1,...,3}{
    \node[] at (2.5*\i*1.25 -1.8*1.25 + 0.25,0) (x\i)  {};
    \node[] at (2.5*\i*1.25 -1.2*1.25 + 0.25,0) (y\i)  {};
    \draw[-stealth, very thick] (x\i) to (y\i);
   }
\end{tikzpicture}}
  \caption{An illustration of a color refinement process.} \label{fig:colorrefinement}
\end{figure}

\noindent \textbf{Color refinement.}
A common algorithm applied when computing the symmetries of a graph is \emph{color refinement}. 
Given a colored graph $G = (V, E, \pi)$, color refinement \emph{refines} the coloring $\pi$ of $G$ into $G' = (V, E, \pi')$. 
Crucially, the automorphism group remains invariant under color refinement, i.e., $\Aut(G) = \Aut(G')$.

We now describe the algorithm.
If two vertices in some color~$X = \pi^{-1}(c)$ have a different number of neighbors in another color~$Y = \pi^{-1}(c')$,
then $X$ can be split by partitioning it according to neighbor counts in $Y$. 
After the split, two vertices have the same color precisely if they had the same color before the split, and they have the same number of neighbors in $Y$.
We repeatedly split classes with respect to other classes until no further splits are possible.  
Figure~\ref{fig:colorrefinement} shows an illustration of the color refinement procedure.
A coloring which does not admit further splits is called \emph{equitable}. 
For a graph $G$, color refinement can be computed in time $\mathcal{O}((|V(G)| + |E(G)|) \log |V(G)|)$ \cite{McKay81practicalgraph,DBLP:conf/wea/Piperno18}.

Let us also recall a different definition for equitable colorings:
A coloring~$\pi$ of a graph is equitable if for all pairs of (not necessarily distinct) color classes~$C_1,C_2$, all vertices in $C_1$ have the same number of neighbors in $C_2$ (i.e., $|N(v)\cap C_2|= |N(v')\cap C_2|$ for all $v,v'\in C_1$).
Given a coloring $\pi$, color refinement computes an equitable refinement $\pi'$, i.e., an equitable coloring $\pi'$ for which $\pi'(v)=\pi'(v')$ implies~$\pi(v)=\pi(v')$. 
In fact, it computes the \emph{coarsest} equitable refinement.

\noindent \textbf{Permutation Groups.}
The \emph{symmetric group} $\Sym(\Omega)$ is the permutation group consisting of all permutations of the set $\Omega$.
A \emph{permutation group} on \emph{domain}~$\Omega$ is a group $\Gamma$ that is a subgroup of~$\Sym(\Omega)$, denoted~$\Gamma \leq \Sym(\Omega)$.
For a subset of the domain $\Omega' \subseteq \Omega$, the \emph{restriction} of~$\Gamma$ to $\Omega'$ is $\Gamma|_{\Omega'} \coloneqq \{\varphi|_{\Omega'}\;|\; \varphi \in \Gamma \}$ (where $\varphi|_{\Omega}$ denotes restricting the domain of $\varphi$ to $\Omega$). The restriction is not necessarily a group since the images need not be in~$\Omega'$.
The \emph{pointwise stabilizer} is the group $\Gamma_{(\Omega')} \coloneqq \{\varphi \in \Gamma \;|\; \forall p \in \Omega': \varphi(p) = p\}$, obtained by fixing all points of $\Omega'$ individually.

Whenever we are dealing with groups, we use a specific, succinct encoding.
Instead of explicitly representing each element of the group, we only store a subset that is sufficient to obtain any other element through composition. 
Formally, let $S$ be a subset of the group $\Gamma$, i.e., $S \subseteq \Gamma$. 
We call $S$ a \emph{generating set} of $\Gamma$ whenever we obtain precisely $\Gamma$ when exhaustively composing elements of $S$.
We write  $\langle S \rangle = \Gamma$.
Moreover, each individual element $\varphi \in S$ can be referred to as a \emph{generator} of $\Gamma$.

We write $\supp(\varphi) \coloneqq \{\omega \;|\; \omega \in \Omega \wedge \varphi(\omega) \neq \omega\}$ for the \emph{support of a map}, meaning points of $\Omega$ not fixed by $\varphi$.
The \emph{support of a group} $\Gamma \in \Sym(\Omega)$ is the union of all supports of elements of $\Gamma$, i.e., $\supp(\Gamma) \coloneqq \{\omega \;|\; \omega \in \Omega \wedge \exists \varphi \in \Gamma: \varphi(\omega) \neq \omega\}$. 

We use the \emph{cycle notation} for permutations $\varphi \colon \Omega \to \Omega$.
The permutation of~$\{1,\ldots,5\}$ given by $1 \mapsto 2, 2 \mapsto 3, 3 \mapsto 1, 4 \mapsto 5, 5 \mapsto 4$ we write as $(1,2,3)(4,5)$.
Note that, for example $(1,2,3)(5,4)$ and~$(3,1,2)(4,5)$ denote the same permutation.
Algorithmically the cycle notation enables us to read and store a permutation $\varphi$ in time $\supp(\varphi)$.  

When considering two permutation groups~$\Gamma$ and~$\Gamma'$ it is possible that the groups are isomorphic as abstract groups but not as permutation groups. For example, if we let the symmetric group~$\Sym(\Omega)$ act component-wise on pairs of elements of~$\Omega$, we obtain a permutation group with domain~$\Omega^2$ that also has~$|\Omega|!$ many elements. In fact this group is isomorphic to~$\Sym(\Omega)$ as an abstract group. We say a group~$\Gamma$ is a \emph{symmetric group in natural action} if the group is~$\Sym(\Omega)$, where~$\Omega$ is the domain of~$\Gamma$.

\noindent \textbf{SAT and Symmetries.} A Boolean satisfiability (SAT) instance $F$ is commonly given in \emph{conjunctive normal form} (CNF), which we denote with
$F = \{(l_{1,1} \vee \cdots{} \vee l_{1,k_1}), \ldots{} ,(l_{m,1} \vee \cdots{} \vee l_{m,k_m})\}$, where each element of $F$ is called a clause. A clause itself consists of a set of \emph{literals}. A literal is either a variable or its negation.
We use $\Var(F) \coloneqq \{v_1, \dots{}, v_n\}$ for the set of \emph{variables} of $F$ and we use $\Lit(F)$ for its literals.

A symmetry, or \emph{automorphism}, of $F$ is a permutation of the literals $\varphi \colon \Lit(F) \to \Lit(F)$ satisfying the following two properties. First, it maps $F$ back to itself, i.e., $\varphi(F) \equiv F$, where $\varphi(F)$ is applied element-wise to the literals in each clause. Here clauses are equivalent, if they are the same when treated as unordered sets of literals, for example $(x \vee y) \equiv (y \vee x)$. Then~$F'\equiv F$ if~$F'$ is obtained from~$F$ by reordering the literals of~$F'$ within the clauses.
Second, for all $l \in \Lit(F)$ it must hold that $\n{\varphi(l)} = \varphi(\n{l})$, i.e., $\varphi$ induces a permutation of the variables. 
For example the permutation mapping~$x_i$ to~$\neg{x_{i+1}}$ and~$x_i$ to~$x_{i+1}$, with indices taken modulo~4, is a symmetry of~$(x_1 \vee \neg x_2 \vee x_3\vee \neg  x_4)\wedge (x_4 \vee \neg x_1 \vee x_2 \vee \neg x_3)$. 
The permutation group of all symmetries of $F$ is $\Aut(F) \leq \Sym(\Lit(F))$.

It is well understood that the symmetries of a SAT formula $F$ can be captured by a graph. 
We call this the \emph{model graph} and denote it with $\mathcal{M}(F)$.\label{model-graph}
While there exists various constructions for the model graphs, we use the following common construction. 
Each literal $l \in \Lit(F)$ is associated with a vertex $l$. 
Each clause $C \in F$ is associated with a vertex~$C$. 
All pairs of literals $l$ and $\n{l}$ are connected by an edge. 
For all literals $l \in C$ of a clause $C$, we connect vertices $l$ and $C$.
Lastly, to distinguish clause vertices from literal vertices, we color all clauses with color $0$ and all literals with color $1$.
As desired, for this graph, $\Aut(F) = \Aut(\mathcal{M}(F))|_{\Lit(F)}$ holds~\cite{DBLP:series/faia/Sakallah21}.

Consider the formula $F_E \coloneqq \{(x_1 \vee \n{y_1}), (x_2 \vee \n{y_2}), (x_3 \vee \n{y_3}), (x_1\vee x_2 \vee x_3\vee z_1\vee z_2)\}$.
Throughout the paper, we use $F_E$ as our running example.
Figure~\ref{fig:model_graph} shows its model graph.
Regarding the symmetries of $F_E$, note that there are symmetries interchanging all of $x_1, x_2, x_3$, all of $y_1, y_2, y_3$ and of the $z_1, z_2$. 
A generating set $S_E$ for $\Aut(F_E)$ is $$S_E = \{(x_1,x_2,x_3)(\n{x_1},\n{x_2},\n{x_3})(y_1,y_2,y_3)(\n{y_1},\n{y_2},\n{y_3}),(x_1,x_2)(\n{x_1},\n{x_2})(y_1,y_2)(\n{y_1},\n{y_2}),(z_1,z_2)\}.$$
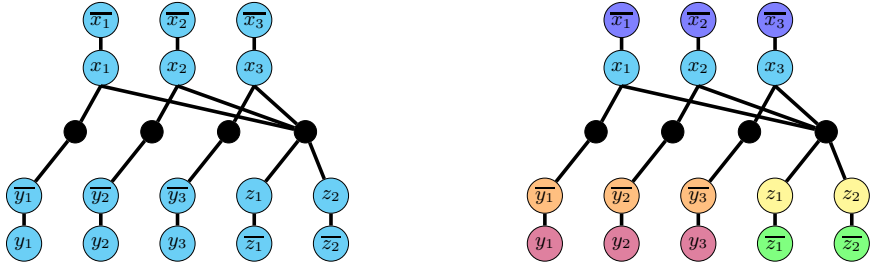
\begin{figure}
	\centering
	\scalebox{0.85}{\begin{tikzpicture}[every node/.style = {draw, circle, fill=black, radius=0.15,minimum size=5.5mm, inner sep=0},yscale=0.5,xscale=0.4]
	\draw (0,0) node [fill=cyan!50] (x1) {\footnotesize $x_1$};
	\draw (0,1.5) node[fill=cyan!50] (nx1) {\footnotesize $\n{x_1}$};
	\draw[ultra thick,-] (x1) -- (nx1);
	
	\draw (3,0)  node[fill=cyan!50] (x2) {\footnotesize $x_2$};
	\draw (3,1.5)  node[fill=cyan!50] (nx2) {\footnotesize $\n{x_2}$};
	\draw[ultra thick,-] (x2) -- (nx2);
	
	\draw (6,0)  node[fill=cyan!50] (x3) {\footnotesize $x_3$};
	\draw (6,1.5)  node[fill=cyan!50] (nx3) {\footnotesize $\n{x_3}$};
	\draw[ultra thick,-] (x3) -- (nx3);
	
	\draw (-4.5+0+1.5,-5.5) node[fill=cyan!50] (y1) {\footnotesize $y_1$};
	\draw (-4.5+0+1.5,-4) node[fill=cyan!50] (ny1) {\footnotesize $\n{y_1}$};
	\draw[ultra thick,-] (y1) -- (ny1);
	
	\draw (-4.5+3+1.5,-5.5)  node[fill=cyan!50] (y2) {\footnotesize $y_2$};
	\draw (-4.5+3+1.5,-4)  node[fill=cyan!50] (ny2) {\footnotesize $\n{y_2}$};
	\draw[ultra thick,-] (y2) -- (ny2);
	
  \draw (-4.5+6+1.5,-5.5)  node[fill=cyan!50] (y3) {\footnotesize $y_3$};
	\draw (-4.5+6+1.5,-4)  node[fill=cyan!50] (ny3) {\footnotesize $\n{y_3}$};
	\draw[ultra thick,-] (y3) -- (ny3);

  \draw (-4.5+6  +1.5*3,-4)  node[fill=cyan!50] (z1) {\footnotesize $z_1$};
	\draw (-4.5+6+1.5*3,-5.5)  node[fill=cyan!50] (nz1) {\footnotesize $\n{z_1}$};
	\draw[ultra thick,-] (z1) -- (nz1);

  \draw (-4.5+6  +1.5*5,-4)  node[fill=cyan!50] (z2) {\footnotesize $z_2$};
	\draw (-4.5+6+1.5*5,-5.5)  node[fill=cyan!50] (nz2) {\footnotesize $\n{z_2}$};
	\draw[ultra thick,-] (z2) -- (nz2);

	\draw (12/4 * 0-1,-2)  node[minimum size=3.5mm] (c1) {};
	\draw (12/4 * 1-1,-2)  node[minimum size=3.5mm] (c2) {};
	\draw (12/4 * 2-1,-2)  node[minimum size=3.5mm] (c3) {};
	\draw (12/4 * 3-1,-2)  node[minimum size=3.5mm] (c4) {};

  \draw[ultra thick,-] (c1) -- (x1.south);
  \draw[ultra thick,-] (c2) -- (x2.south);
  \draw[ultra thick,-] (c3) -- (x3.south);
  \draw[ultra thick,-] (c1) -- (ny1);
  \draw[ultra thick,-] (c2) -- (ny2);
  \draw[ultra thick,-] (c3) -- (ny3);

  \draw[ultra thick,-] (c4) -- (x1.south);
  \draw[ultra thick,-] (c4) -- (x2.south);
  \draw[ultra thick,-] (c4) -- (x3.south);
  \draw[ultra thick,-] (c4) -- (z1);
  \draw[ultra thick,-] (c4) -- (z2);
	\end{tikzpicture}}
  \hspace{2cm}
  \scalebox{0.85}{
    \begin{tikzpicture}[every node/.style = {draw, circle, fill=black, radius=0.15,minimum size=5.5mm, inner sep=0},yscale=0.5,xscale=0.4]
    \draw (0,0) node [fill=cyan!50] (x1) {\footnotesize $x_1$};
    \draw (0,1.5) node[fill=blue!50] (nx1) {\footnotesize $\n{x_1}$};
    \draw[ultra thick,-] (x1) -- (nx1);
    
    \draw (3,0)  node[fill=cyan!50] (x2) {\footnotesize $x_2$};
    \draw (3,1.5)  node[fill=blue!50] (nx2) {\footnotesize $\n{x_2}$};
    \draw[ultra thick,-] (x2) -- (nx2);
    
    \draw (6,0)  node[fill=cyan!50] (x3) {\footnotesize $x_3$};
    \draw (6,1.5)  node[fill=blue!50] (nx3) {\footnotesize $\n{x_3}$};
    \draw[ultra thick,-] (x3) -- (nx3);
    
    \draw (-4.5+0+1.5,-5.5) node[fill=purple!50] (y1) {\footnotesize $y_1$};
    \draw (-4.5+1.5,-4) node[fill=orange!50] (ny1) {\footnotesize $\n{y_1}$};
    \draw[ultra thick,-] (y1) -- (ny1);
    
    \draw (-4.5+3+1.5,-5.5)  node[fill=purple!50] (y2) {\footnotesize $y_2$};
    \draw (-4.5+4.5,-4)  node[fill=orange!50] (ny2) {\footnotesize $\n{y_2}$};
    \draw[ultra thick,-] (y2) -- (ny2);
    
    \draw (-4.5+6+1.5,-5.5)  node[fill=purple!50] (y3) {\footnotesize $y_3$};
    \draw (-4.5+7.5,-4)  node[fill=orange!50] (ny3) {\footnotesize $\n{y_3}$};
    \draw[ultra thick,-] (y3) -- (ny3);
  
    \draw (-4.5+6  +1.5*3,-4)  node[fill=yellow!50] (z1) {\footnotesize $z_1$};
    \draw (-4.5+6+1.5*3,-5.5)  node[fill=green!50] (nz1) {\footnotesize $\n{z_1}$};
    \draw[ultra thick,-] (z1) -- (nz1);
  
    \draw (-4.5+6  +1.5*5,-4)  node[fill=yellow!50] (z2) {\footnotesize $z_2$};
    \draw (-4.5+6+1.5*5,-5.5)  node[fill=green!50] (nz2) {\footnotesize $\n{z_2}$};
    \draw[ultra thick,-] (z2) -- (nz2);
  
    \draw (12/4 * 0-1,-2)  node[minimum size=3.5mm] (c1) {};
	\draw (12/4 * 1-1,-2)  node[minimum size=3.5mm] (c2) {};
	\draw (12/4 * 2-1,-2)  node[minimum size=3.5mm] (c3) {};
	\draw (12/4 * 3-1,-2)  node[minimum size=3.5mm] (c4) {};
  
    \draw[ultra thick,-] (c1) -- (x1.south);
    \draw[ultra thick,-] (c2) -- (x2.south);
    \draw[ultra thick,-] (c3) -- (x3.south);
    \draw[ultra thick,-] (c1) -- (ny1);
    \draw[ultra thick,-] (c2) -- (ny2);
    \draw[ultra thick,-] (c3) -- (ny3);
  
    \draw[ultra thick,-] (c4) -- (x1.south);
    \draw[ultra thick,-] (c4) -- (x2.south);
    \draw[ultra thick,-] (c4) -- (x3.south);
    \draw[ultra thick,-] (c4) -- (z1);
    \draw[ultra thick,-] (c4) -- (z2);
    \end{tikzpicture}
    }
	\caption{Example model graph $\mathcal{M}(F_E) = \mathcal{M}(\{(x_1 \vee \n{y_1}), (x_2 \vee \n{y_2}), (x_3 \vee \n{y_3}), (x_1\vee x_2 \vee x_3\vee z_1\vee z_2)\})$. The automorphisms of the graph correspond to the automorphisms of the formula. The coloring on the right side shows the orbit partition.}
	\label{fig:model_graph}
\end{figure}

\noindent \textbf{Orbits.}
Given a permutation group $\Gamma \leq \Sym(\Omega)$, we denote with $\omega^\Gamma \subseteq \Omega$ the \emph{orbit} of a point $\omega \in \Omega$. 
That is, an element $\omega' \in \Omega$ is in $\omega^\Gamma$ whenever there is a $\varphi \in \Gamma$ with $\varphi(\omega') = \omega$. 

The orbits of our example $\Aut(F_E)$ are shown in Figure~\ref{fig:model_graph}, e.g., the orbit $\{\n{z_1}, \n{z_2}\}$ is green. 
\section{SAT-Symmetry Computational Setting} \label{sec:setting}
Let us describe the computational setting in which our new algorithms operate. Since we want the theoretical runtimes to reflect more closely the runtimes in practice, there are two important differences compared to the traditional computational group theory setting. 
These differences are in the measure of runtime as well as in the format of the input.

\noindent \textbf{Joint Graph/Group Pairs.} 
Typically, algorithms in computational group theory dealing with permutation groups assume as their input a generating set of permutations $S$ of a group $\Gamma = \langle S \rangle$.
While this is certainly a natural setting when discussing algorithms for groups in general, 
in our setting this input format disregards further information that is readily available. Therefore, we require that algorithms in the SAT-symmetry interface have access to more information about the input group.
Specifically, we may require that the input consists of \emph{both} a generating set $S$ and a graph $G$ with $\langle S \rangle = \Aut(G)$.
We call this a \emph{joint graph/group pair} $(G, S)$. 
For our SAT context, we may moreover assume that the SAT formula $F$ with $\mathcal{M}(F) = G$ is available, whenever necessary.

\noindent \textbf{Instance-Linear Runtime.} 
In computational group theory, given a generating set $S$ for a permutation group $\langle S \rangle \leq \Sym(\Omega)$, a runtime of $\mathcal{O}(|S||\Omega|)$ is typically considered linear time~\cite{seress_2003}.
This is however only a very crude upper bound when seen in terms of the actual encoding size of a given generating set. In particular, when generators are sparse, as is common in SAT \cite{DBLP:series/faia/Sakallah21, DBLP:conf/dac/DargaLSM04},
linear time in this sense is \emph{not} necessarily linear time in the encoding size, which is what we would use in a graph-theoretic or SAT context.

Specifically, we are interested in measuring the runtime of algorithms relative to the encoding size of a generating set given in a sparse format.
Therefore, we define the encoding size of a generating set $S$ as $\enc(S) \coloneqq \Sigma_{p \in S} |\supp(p)|$.

In particular, given a SAT formula $F$, graph $G = (V, E)$, and generating set $S$, the goal is to have algorithms that (ideally) run in time linear in $|F| + |V| + |E| + \enc(S)$.
In order to not confuse the ``types of linear time'', we refer to such algorithms as \emph{instance-linear}. Analogously, an algorithm has \emph{instance-quasi-linear} time if it runs in time~$\mathcal{O}((|F| + |V| + |E| + \enc(S))\cdot (\log(|F| + |V| + |E| + \enc(S)))^c)$ for some constant~$c$.

\noindent \textbf{Illustrative Examples.}
The task of computing the orbits is an excellent example demonstrating the usefulness instance-quasi-linear time. As transitive closure, we can find the orbit~$\Delta$ of an element in time $\mathcal{O}(|\Delta||S|)$~\cite{seress_2003}. 
However, with instance-quasi-linear time in mind, we quickly arrive at an algorithm to compute the entire orbit partition in time~$\mathcal{O}(\enc(S)\cdot \alpha(\enc(S)))$ using a union-find data structure, where $\alpha$ is the inverse Ackermann function. The inverse Ackermann function exhibits substantially slower growth than $\log(n)$.

Furthermore, having access to the graph of a graph/group pair $(G,S)$ gives a significant advantage in what is algorithmically possible. A good example of the difference
is that testing membership $\varphi \in \langle S \rangle$ is much easier for the graph/group pair:
testing $\varphi(G) = G$ (which is true if and only if $\varphi \in \langle S \rangle$) can be done in instance-linear time.
However, testing $\varphi \in \langle S \rangle$ without access to the graph is much more involved. 
The best known method for the latter involves computing a strong generating set, corresponding base and Schreier table~\cite{seress_2003}, followed by an application of the fundamental sifting algorithm~\cite{seress_2003}. 
Even performing only the last step of this process (sifting) is not guaranteed to be in instance-linear time.

\section{Favorable Group Structures in SAT} \label{sec:satandsym}
We now propose problems which should be solved within the SAT-symmetry computational setting.   
We analyze heuristics used in advanced symmetry exploitation algorithms \cite{DBLP:conf/sat/Devriendt0BD16, DBLP:journals/mpc/PfetschR19, Grayland2009MinimalOC}, extracting precise formal definitions.

\noindent \textbf{Disjoint Direct Decomposition.}
Following~\cite{DBLP:journals/jsc/ChangJ22}, we say a direct product of a permutation group $\Gamma = \Gamma_1 \times \Gamma_2 \times \cdots \times \Gamma_r$ is a \emph{disjoint direct decomposition} of $\Gamma$, whenever all $\Gamma_i$ have pairwise disjoint supports. We call $\Gamma_i$ a \emph{factor} of the disjoint direct decomposition of $\Gamma$.
A disjoint direct decomposition is \emph{finest}, if we cannot decompose any factor further into a non-trivial disjoint direct decomposition.
A recent algorithm solves the problem of computing the finest disjoint direct decomposition
for permutation groups in polynomial-time~\cite{DBLP:journals/jsc/ChangJ22}.

In our running example, the finest disjoint direct decomposition of $\Aut(F_E)$ splits the group into a subgroup $H_1$ permuting only the $x_i$ and $y_i$ variables, and a subgroup $H_2$ permuting $z_1$ and~$z_2$.
Indeed, setting $H_1 = \langle\{(x_1,x_2,x_3)(\n{x_1},\n{x_2},\n{x_3})(y_1,y_2,y_3)(\n{y_1},\n{y_2},\n{y_3}),\allowbreak (x_1,x_2)(\n{x_1},\n{x_2})(y_1,y_2)(\n{y_1},\n{y_2})\} \rangle$, and 
$H_2 = \langle\{(z_1,z_2)\}\rangle$, we have~$\Aut(F_E) = H_1 \times H_2$.

Computing a disjoint direct decomposition is a typical routine in symmetry exploitation tools \cite{DBLP:conf/sat/Devriendt0BD16, DBLP:journals/mpc/PfetschR19, Grayland2009MinimalOC}.
It allows for separate treatment of each factor of the decomposition.
The heuristics in use today do not guarantee that the decomposition is the finest disjoint direct decomposition:
indeed, the heuristics of the tools mentioned above assume that the given generating sets are already \emph{separable} \cite{DBLP:journals/jsc/ChangJ22}. 
This means it is assumed, that every generator $\varphi \in S$ only operates on one factor of the disjoint direct decomposition $\Gamma_1 \times \Gamma_2 \times \dots\times \Gamma_r$. 
Formally, this means for each $\varphi \in S$ there is one $i \in \{1, \dots{}, r\}$ for which $\supp(\varphi) \cap \supp(\Gamma_i) \neq \emptyset$ holds, and for all $j \in \{1, \dots{}, r\}, j \neq i$ it holds that $\supp(\varphi) \cap \supp(\Gamma_i) = \emptyset$.
For example, the generating set $S_E$ we gave for $\Aut(F_E)$ is separable. 

It is not known how often generating sets given for graphs of SAT formulas are separable for a given symmetry detection tool, or in particular after reducing the domain to the literals of the SAT formula.
It is however obvious that for the most advanced general-purpose symmetry detection tools, \Traces{} and \dejavu{}, generators are not separable with very high probability due to the use of randomly selected generators \cite{DBLP:journals/jsc/ChangJ22}.

\noindent \textbf{Row Interchangeability.} 
We now discuss the concept of row interchangeability \cite{DBLP:conf/sat/Devriendt0BD16, DBLP:conf/cp/FlenerFHKMPW02, DBLP:journals/mpc/PfetschR19}.
Let $F$ be a SAT formula.
Let $M$ be a variable matrix $M \colon \{1, \dots{}, r\} \times \{1, \dots{}, c\} \to \Var(F)$.
We denote the entries of $M$ with $x_{i,j}$ where $i \in \{1, \dots{}, r\}$ and $j \in \{1, \dots{}, c\}$. We define the shorthand $\supp(M) \coloneqq \{x_{i,j} \mid x_{i,j} \in M\} \cup \{\n{x_{i,j}} \mid x_{i,j} \in M\}$. The set $\supp(M)$ denotes all the literals involved with the matrix $M$.
We say $F$ exhibits \emph{row interchangeability} if there exists a matrix $M$ such that for every permutation $\varphi \in \Sym(\{1, \dots{}, r\})$,
for the induced literal permutation $\varphi' \colon \Lit(F)|_{\supp(M)} \to \Lit(F)|_{\supp(M)}$ given by $x_{i,j} \mapsto x_{\varphi(i),j}, \neg x_{i,j} \mapsto \neg x_{\varphi(i),j}$ it holds that $\varphi' \in \Aut(F)|_{\supp(M)}$.
Indeed, if this is the case, we can observe that the matrix $M$ describes a subgroup of $\Aut(F)$ consisting of~$\{\pi\in \Aut(F) \mid \exists\varphi \in \Sym(\{1, \dots{}, r\}): \varphi'=\pi|_{\supp(M)}\}$. We denote this group by~$H_M \leq \Aut(F)$.

A crucial fact is that for $H_M|_{\supp(M)}$, linear-sized complete symmetry breaking is available \cite{DBLP:conf/sat/Devriendt0BD16, DBLP:conf/cp/FlenerFHKMPW02}.
As is also in part discussed in \cite{DBLP:conf/sat/Devriendt0BD16, DBLP:conf/cp/FlenerFHKMPW02}, we observe that the complete symmetry breaking for $H_M$ is most effective whenever $H_M$ is the only action on $\supp(M)$ in $\Aut(F)$, or more precisely, $\Aut(F)|_{\supp(M)} = H_M|_{\supp(M)}$. 
In this case, we call $H_M$ an \emph{elementary} row interchangeability subgroup.
Otherwise, there are non-trivial symmetries $\varphi \in \Aut(F)|_{\supp(M)}$ with $\varphi \not\in H_M|_{\supp(M)}$ or~$\supp(M)$ is not a union of orbits.
Indeed, in this case, the complete symmetry breaking of $H_M$ might make it more difficult to break such overlapping symmetries $\varphi$: for example, if two row interchangeability subgroups $H_M$ and $H_{M'}$ overlap, i.e., $\supp(M) \cap \supp(M') \neq \emptyset$, complete symmetry breaking can only be guaranteed for one of them using the technique of \cite{DBLP:conf/sat/Devriendt0BD16}. 

Whenever $H_M$ is an elementary row interchangeability subgroup, the situation is much clearer: we can produce a linear-sized complete symmetry breaking formula and this covers at least all symmetries on the literals $\supp(M)$. 
In this paper, we therefore focus on computing elementary row interchangeability groups.

Let us consider $F_E$ again: for the matrix
$M \coloneqq \begin{bmatrix}
  x_1 & x_2 & x_3 \\
  y_1 & y_2 & y_3 
\end{bmatrix}$
there is indeed a row interchangeability subgroup.
(Recall that the group $H_M$ permutes positive and negative literals of variables appearing in~$M$.) For this example, $H_M$ is both an elementary row interchangeability group and a factor in the finest direct disjoint decomposition of $\Aut(F_E)$.

\noindent \textbf{Row Interchangeability and Equivalent Orbits.} 
We now describe the matrix of elementary row interchangeability groups in more group-theoretic terms.
We first define the notion of equivalent orbits:
\begin{definition}[Equivalent orbits (see {\cite[Subsection 6.1.2]{seress_2003}})] Two orbits $\Delta_1, \Delta_2$ are equivalent, if and only if there is a bijection $b \colon \Delta_1 \to \Delta_2$ such that for all $\varphi \in \Gamma$ and $\delta \in \Delta_1$, $\varphi(b(\delta)) = b(\varphi(\delta))$.
\end{definition}
We write $\Delta_1 \equiv \Delta_2$ to indicate orbits~$\Delta_1$ and~$\Delta_2$ are equivalent.
It is easy to see this indeed defines an equivalence relation on the orbits \cite{seress_2003}. 

We observe that if a row interchangeability subgroup $H_M$ is elementary, each row of the matrix $M$ is an orbit of $\Aut(F)$.
Since all rows are moved simultaneously in the same way, we remark that rows of $M$ are precisely equivalent orbits with a natural symmetric action:
\begin{lemma} Let $H_M$ be a row interchangeability subgroup of $\Gamma = \Aut(\mathcal{M}(F))$, and let $\Delta_i = (x_1, \dots{}, x_c)$ denote a row of $M$.
  $H_M$ is an elementary row interchangeability subgroup if and only if all of the following hold: (1) $\Delta_i$ is an orbit with a natural symmetric action in $\Gamma$. (2) For every other row $\Delta_j$ of $M$, it holds that $\Delta_i \equiv \Delta_j$. (3) For $\n{\Delta_i} = (\n{x_1}, \dots{}, \n{x_c})$, $\n{\Delta_i}$ is also an orbit with $\n{\Delta_i} \equiv \Delta_i$. 
\label{lem:columninterchangeeq}
\end{lemma}
There is an exact algorithm which computes equivalent orbits in essentially quadratic runtime~\cite{seress_2003}. Again, runtimes are difficult to compare due to different pre-conditions in \cite{seress_2003}. 
In any case, the algorithm for equivalent orbits depends on computing a base and strong generating set, which is too slow from our perspective. 

We may split detecting elementary row interchangeability groups into detecting \emph{natural symmetric action} on the orbits, followed by computing \emph{equivalent orbits}.
We now turn to solving the problems defined above in the computational setting of the SAT-symmetry interface.
Specifically, we propose algorithms for the finest disjoint direct decomposition (Section~\ref{sec:decompose}), natural symmetric action (Section~\ref{sec:snaction}), and equivalent orbits (Section~\ref{sec:rowinterchange}).

\section{Finest Disjoint Direct Decomposition} \label{sec:decompose}
Having established the problems we want to address, we now turn to presenting suitable algorithms in the SAT-symmetry computational setting. In particular, recall that we want to make use of joint graph/group pairs in order to state algorithms that run in instance-quasi-linear time. We begin by computing the finest disjoint direct decomposition.

Specifically, given a joint graph/group pair $(G,S)$, our aim is to compute the finest disjoint direct decomposition of the group $\langle S \rangle$.
Our proposed algorithm, given the orbits, can do so in instance-linear time.
The disjoint direct decomposition of a group allows us to separately treat each factor of the decomposition in symmetry exploitation or other consecutive algorithms. 

To simplify the discussion, we assume the graph $G$ to be undirected. However, the procedure generalizes to both directed and even edge-colored graphs.

\textbf{Orbit Graph.}
We describe the \emph{orbit graph}, which can be constructed from $(G,S)$.
We are particularly interested in the connected components of the orbit graph, which turn out to correspond exactly to the factors of the finest disjoint direct decomposition.

First, note that the orbit partition $\pi$ of $\langle S \rangle$ can be viewed as a vertex coloring of the graph $G$, assigning to every vertex its orbit. 
We consider the graph $G' = (V(G), E(G), \pi)$, i.e., $G$ colored with its orbit partition (see Figure~\ref{fig:runningexampledecompose}, left).

We call two distinct orbits $\Delta, \Delta'$ \emph{homogeneously connected}, whenever either all vertices $v \in \Delta$ are adjacent to all vertices of $\Delta'$, or there is no edge with endpoints both in~$\Delta$ and~$\Delta'$. 
Indeed, we could ``flip edges'' between homogeneously connected orbits such that they all become disconnected, without changing the automorphism group (see Figure~\ref{fig:runningexampledecompose}, middle).

We now give the formal definition of the orbit graph.
The orbit graph is essentially an adapted version of the so-called flipped quotient graph (see \cite{DBLP:conf/mfcs/KieferSS15} for a discussion).
The vertex set of the orbit graph is the set of orbits of $\langle S \rangle$, i.e., $\{\pi^{-1}(v) \; | \; v \in V(G')\}$.
Two orbits $\Delta, \Delta'$ are adjacent in the orbit graph if and only if the orbits are \emph{not} homogeneously connected in the original graph $G$ (see Figure~\ref{fig:runningexampledecompose}, right).

\SetKwFor{For}{for (}{)}{}
\begin{algorithm}[t]
  \SetAlgoLined
	\SetAlgoNoEnd
	\caption{Compute the orbit graph.} \label{alg:orbitgraph}
	\Fn{\OrbitGraph{$G'$}}{
		\SetKwInOut{Input}{Input}
		\SetKwInOut{Output}{Output}
		\Input{graph $G' = (V, E, \pi)$, where $\pi$ is the orbit partition of $\Aut(G')$}
		\Output{orbit graph $G_O$}
    initialize integer array $A$ of size $|V|$ with all $0$\;
    initialize empty list $W$\; 
    initialize empty graph $G_O$\;
    $V(G_O) \coloneqq \pi(V)$\;
    \For{$\Delta \in \pi(V)$}{
      pick an arbitrary $v \in \Delta$\;
      \For{$(v, v') \in E$}{
        increment $A[\pi(v')]$\; 
        add $\pi(v')$ to $W$\;
      }
      \For{$\Delta' \in W$}{
        \lIf{$A[\Delta'] > 0$ and $A[\Delta'] < |\pi^{-1}(\Delta')|$\label{alg:line:homogeneoustest}}{
          add edge $(\Delta, \Delta')$ in $G_O$ \label{alg:line:connectedge}
        }
        $A[\Delta']$ := $0$\;
      }
    }
    
		\Return{$G_O$}\;
	}
\end{algorithm}

\emph{Description of Algorithm~\ref{alg:orbitgraph}}. Algorithm~\ref{alg:orbitgraph} describes how to compute the orbit graph from $G'$. The algorithm first initializes the graph $G_O$ with a vertex set that contains exactly one vertex for each orbit of~$G'$.
It then counts for each orbit~$\Delta$, how many neighbors a vertex~$v\in \Delta$ has in the other orbits. Since~$\Delta$ is an orbit, this number is the same for all vertices, so it suffices to compute this for one~$v\in \Delta$.
Finally, the algorithm checks to which other colors the vertex~$v$ and thus the orbit~$\Delta$ is \emph{not} homogeneously connected (Line~\ref{alg:line:homogeneoustest}).
If~$\Delta$ and~$\Delta'$ are not homogeneously connected, the edge~$(\Delta,\Delta')$ is added to the orbit graph.  

\emph{Remark on the runtime of Algorithm~\ref{alg:orbitgraph}.} Using appropriate data structures for graphs (adjacency lists) and colorings (see \cite{DBLP:journals/jsc/McKayP14}, which in particular includes efficient ways to compute $|\pi^{-1}(C')|$), the algorithm can be implemented in instance-linear time.

\pgfdeclarelayer{background}
\pgfdeclarelayer{foreground}
\pgfsetlayers{background,main,foreground}

\begin{figure}[t]
	\centering
  \scalebox{0.85}{
	\begin{tikzpicture}[every node/.style = {draw, circle, fill=black, radius=0.15,minimum size=5.5mm, inner sep=0},yscale=0.5,xscale=0.4]
	\draw (0,0) node [fill=cyan!50] (x1) {\footnotesize $x_1$};
	\draw (0,1.5) node[fill=blue!50] (nx1) {\footnotesize $\n{x_1}$};
	\draw[ultra thick,-] (x1) -- (nx1);
	
	\draw (3,0)  node[fill=cyan!50] (x2) {\footnotesize $x_2$};
	\draw (3,1.5)  node[fill=blue!50] (nx2) {\footnotesize $\n{x_2}$};
	\draw[ultra thick,-] (x2) -- (nx2);
	
	\draw (6,0)  node[fill=cyan!50] (x3) {\footnotesize $x_3$};
	\draw (6,1.5)  node[fill=blue!50] (nx3) {\footnotesize $\n{x_3}$};
	\draw[ultra thick,-] (x3) -- (nx3);
	
	\draw (-4.5+0+1.5,-5.5) node[fill=purple!50] (y1) {\footnotesize $y_1$};
	\draw (-4.5+1.5,-4) node[fill=orange!50] (ny1) {\footnotesize $\n{y_1}$};
	\draw[ultra thick,-] (y1) -- (ny1);
	
	\draw (-4.5+3+1.5,-5.5)  node[fill=purple!50] (y2) {\footnotesize $y_2$};
	\draw (-4.5+4.5,-4)  node[fill=orange!50] (ny2) {\footnotesize $\n{y_2}$};
	\draw[ultra thick,-] (y2) -- (ny2);
	
  \draw (-4.5+6+1.5,-5.5)  node[fill=purple!50] (y3) {\footnotesize $y_3$};
	\draw (-4.5+7.5,-4)  node[fill=orange!50] (ny3) {\footnotesize $\n{y_3}$};
	\draw[ultra thick,-] (y3) -- (ny3);

  \draw (-4.5+6  +1.5*3,-4)  node[fill=yellow!50] (z1) {\footnotesize $z_1$};
	\draw (-4.5+6+1.5*3,-5.5)  node[fill=green!50] (nz1) {\footnotesize $\n{z_1}$};
	\draw[ultra thick,-] (z1) -- (nz1);

  \draw (-4.5+6  +1.5*5,-4)  node[fill=yellow!50] (z2) {\footnotesize $z_2$};
	\draw (-4.5+6+1.5*5,-5.5)  node[fill=green!50] (nz2) {\footnotesize $\n{z_2}$};
	\draw[ultra thick,-] (z2) -- (nz2);

	\draw[] (12/4 * 0-1,-2)  node[fill=black,minimum size=4.5mm] (c1) {\color{white}$c_1$};
	\draw (12/4 * 1-1,-2)  node[fill=black,minimum size=4.5mm] (c2) {\color{white}$c_2$};
	\draw (12/4 * 2-1,-2)  node[fill=black,minimum size=4.5mm] (c3) {\color{white}$c_3$};
	\draw (12/4 * 3-1,-2)  node[fill=gray,minimum size=4.5mm] (c4) {\color{white}$c_4$};

  \draw[ultra thick,-] (c1) -- (x1.south);
  \draw[ultra thick,-] (c2) -- (x2.south);
  \draw[ultra thick,-] (c3) -- (x3.south);
  \draw[ultra thick,-] (c1) -- (ny1);
  \draw[ultra thick,-] (c2) -- (ny2);
  \draw[ultra thick,-] (c3) -- (ny3);

  \draw[ultra thick,-] (c4) -- (x1.south);
  \draw[ultra thick,-] (c4) -- (x2.south);
  \draw[ultra thick,-] (c4) -- (x3.south);
  \draw[ultra thick,-] (c4) -- (z1);
  \draw[ultra thick,-] (c4) -- (z2);
	\end{tikzpicture}
  }
  \hspace{1cm}
  \scalebox{0.85}{
	\begin{tikzpicture}[every node/.style = {draw, circle, fill=black, radius=0.15,minimum size=5.5mm, inner sep=0},yscale=0.5,xscale=0.4]
	\draw (0,0) node [fill=cyan!50] (x1) {\footnotesize $x_1$};
	\draw (0,1.5) node[fill=blue!50] (nx1) {\footnotesize $\n{x_1}$};
	\draw[ultra thick,-] (x1) -- (nx1);
	
	\draw (3,0)  node[fill=cyan!50] (x2) {\footnotesize $x_2$};
	\draw (3,1.5)  node[fill=blue!50] (nx2) {\footnotesize $\n{x_2}$};
	\draw[ultra thick,-] (x2) -- (nx2);
	
	\draw (6,0)  node[fill=cyan!50] (x3) {\footnotesize $x_3$};
	\draw (6,1.5)  node[fill=blue!50] (nx3) {\footnotesize $\n{x_3}$};
	\draw[ultra thick,-] (x3) -- (nx3);
	
	\draw (-4.5+0+1.5,-5.5) node[fill=purple!50] (y1) {\footnotesize $y_1$};
	\draw (-4.5+1.5,-4) node[fill=orange!50] (ny1) {\footnotesize $\n{y_1}$};
	\draw[ultra thick,-] (y1) -- (ny1);
	
	\draw (-4.5+3+1.5,-5.5)  node[fill=purple!50] (y2) {\footnotesize $y_2$};
	\draw (-4.5+4.5,-4)  node[fill=orange!50] (ny2) {\footnotesize $\n{y_2}$};
	\draw[ultra thick,-] (y2) -- (ny2);
	
  \draw (-4.5+6+1.5,-5.5)  node[fill=purple!50] (y3) {\footnotesize $y_3$};
	\draw (-4.5+7.5,-4)  node[fill=orange!50] (ny3) {\footnotesize $\n{y_3}$};
	\draw[ultra thick,-] (y3) -- (ny3);

  \draw (-4.5+6  +1.5*3,-4)  node[fill=yellow!50] (z1) {\footnotesize $z_1$};
	\draw (-4.5+6+1.5*3,-5.5)  node[fill=green!50] (nz1) {\footnotesize $\n{z_1}$};
	\draw[ultra thick,-] (z1) -- (nz1);

  \draw (-4.5+6  +1.5*5,-4)  node[fill=yellow!50] (z2) {\footnotesize $z_2$};
	\draw (-4.5+6+1.5*5,-5.5)  node[fill=green!50] (nz2) {\footnotesize $\n{z_2}$};
	\draw[ultra thick,-] (z2) -- (nz2);

	\draw[] (12/4 * 0-1,-2)  node[fill=black,minimum size=4.5mm] (c1) {\color{white}$c_1$};
	\draw (12/4 * 1-1,-2)  node[fill=black,minimum size=4.5mm] (c2) {\color{white}$c_2$};
	\draw (12/4 * 2-1,-2)  node[fill=black,minimum size=4.5mm] (c3) {\color{white}$c_3$};
	\draw (12/4 * 3-1,-2)  node[fill=gray,minimum size=4.5mm] (c4) {\color{white}$c_4$};

  \draw[ultra thick,-] (c1) -- (x1.south);
  \draw[ultra thick,-] (c2) -- (x2.south);
  \draw[ultra thick,-] (c3) -- (x3.south);
  \draw[ultra thick,-] (c1) -- (ny1);
  \draw[ultra thick,-] (c2) -- (ny2);
  \draw[ultra thick,-] (c3) -- (ny3);
	\end{tikzpicture}
  }
  \hspace{1cm}
  \scalebox{0.85}{
	\begin{tikzpicture}[every node/.style = {draw, rectangle, fill=black, radius=0.15,minimum size=5.5mm, inner sep=0},yscale=0.5,xscale=0.4]
	\draw (0,0) node [fill=cyan!50] (x1) {\footnotesize $\{x_1,x_2,x_3\}$};
	\draw (0,1.5) node[fill=blue!50] (nx1) {\footnotesize $\{\n{x_1},\n{x_2},\n{x_3}\}$};
	\draw[ultra thick,-] (x1) -- (nx1);
	
	\draw (0,-5.5) node[fill=purple!50] (y1) {\footnotesize $\{y_1,y_2,y_3\}$};
	\draw (0,-4) node[fill=orange!50] (ny1) {\footnotesize $\{\n{y_1},\n{y_2},\n{y_3}\}$};
	\draw[ultra thick,-] (y1) -- (ny1);

  \draw (-4.5+6  +1.5*3,-4)  node[fill=yellow!50] (z1) {\footnotesize $\{z_1,z_2\}$};
	\draw (-4.5+6+1.5*3,-5.5)  node[fill=green!50] (nz1) {\footnotesize $\{\n{z_1},\n{z_2}\}$};
	\draw[ultra thick,-] (z1) -- (nz1);

	\draw[] (0,-2)  node[fill=black] (c1) {\color{white}$\{c_1,c_2,c_3\}$};
	\draw (-4.5+6+1.5*3,-2)  node[fill=gray] (c4) {\color{white}$\{c_4\}$};

  \draw[ultra thick,-] (c1) -- (x1);
  \draw[ultra thick,-] (c1) -- (ny1);
	\end{tikzpicture}
  }
	\caption{Model graph of $F_E$ colored with its orbit partition on the left. The corresponding graph with flipped edges is in the middle, which disconnects parts of the graph. On the right the orbit graph is shown, whose $3$ connected components correspond to the factors of the finest disjoint direct decomposition.}
	\label{fig:runningexampledecompose}
\end{figure}
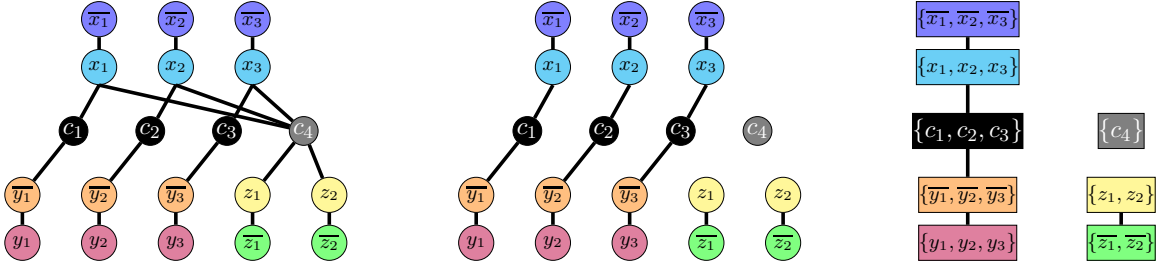

\textbf{Orbit Graph to Decomposition.}
Indeed, the connected components of the orbit graph represent precisely the factors of the finest disjoint direct decomposition of the automorphism group of the graph: 
\begin{restatable}{lemma}{decompositionlemma} Let $\Gamma = \Aut(G)$. The vertices represented by a connected component of the orbit graph of $G$ are all in the same factor of the finest direct disjoint decomposition of $\Gamma$ and vice versa.
\end{restatable}
\begin{proof}
  Consider two orbits $\Delta_1, \Delta_2$ of $\Gamma$ in different factors of any direct disjoint decomposition. 
  Towards a contradiction, assume $\Delta_1, \Delta_2$ are not homogeneously connected.
  Note that naturally, the orbit coloring is equitable.
  Since the orbit coloring is equitable, the connection must be regular, i.e., each vertex of $\Delta_1$ has $d_1$ neighbors in $\Delta_2$, and every vertex of $\Delta_2$ has $d_2$ neighbors in $\Delta_1$ for some integers~$d_1,d_2$.
  However, $0 < d_1 < |\Delta_2|$ and $0 < d_2 < |\Delta_1|$ hold.
  
  Let us now fix a point $\delta \in \Delta_1$, i.e., consider the point stabilizer $\Gamma_{(\delta)}$. If two orbits are in different factors of a direct disjoint decomposition, fixing a point of $\Delta_1$ must not change the group action on $\Delta_2$.
  In particular, $\Delta_2$ must be an orbit of $\Gamma_{(\delta)}$.
  However, $\delta$ is adjacent to some vertex~$\delta'\in \Delta_2$ and non-adjacent to some vertex~$\delta''\in \Delta_2$ (see Figure~\ref{fig:orbithomogeneous} for an illustration).
  Having fixed $\delta$, we can therefore not map~$d'$ to~$d''$. This contradicts the assumption that $\Delta_1$ and $\Delta_2$ are in different factors of a direct disjoint decomposition.
  Hence, orbits in different factors of a direct disjoint decomposition must be homogeneously connected in $G$, i.e., non-adjacent in the orbit graph.

  Now assume $\Delta_1$ and $\Delta_k$ are in the same component in the orbit graph.
  Then, there must be a path of orbits $\Delta_1, \Delta_2, \dots{}, \Delta_k$ where each $\Delta_i, \Delta_{i+1}$ is \emph{not} homogeneously connected.
  In this case, we know for each~$i\in \{1,\ldots,k-1\}$ that $\Delta_i$ and $\Delta_{i+1}$ must be in the same factor of any disjoint direct decomposition.
  Therefore, $\Delta_1$ and $\Delta_k$ must be in the same factor of every disjoint direct decomposition.

  On the other hand, if $\Delta_1$ and $\Delta_k$ are in different components in the orbit graph, they are in different factors of the finest disjoint direct decomposition.
\end{proof}
\begin{figure}
  \centering
  \scalebox{0.9}{\begin{tikzpicture}[scale=1.5,yscale=0.8]

  \foreach \i in {0,...,4}{
    \node[draw, fill=cyan!50, circle] at (0.5*\i, 0.5)  (a\i)  {};
  }

  \foreach \i in {0,...,4}{
    \node[draw, fill=purple!50, circle] at (0.5*\i, -0.5)  (b\i)  {};
  }

  \begin{pgfonlayer}{background}
  \foreach \i in {0,...,4}{
    \pgfmathtruncatemacro{\il}{mod(\i+1,5)}
    \draw[ultra thick,] ($ (a\i.south) + (0,0.01) $)-- ($ (b\il.north) - (0,0.025) $);
    \draw[ultra thick] ($ (a\i.south) + (0,0.01) $) -- ($ (b\i.north) - (0,0.025) $);
  }
\end{pgfonlayer}
\end{tikzpicture}}
\hspace{2cm}
\scalebox{0.9}{\begin{tikzpicture}[scale=1.5,yscale=0.8]
  
  \node[draw, fill=red!50, circle] at (0.5*0, 0.5)  (a0)  {};
   \foreach \i in {1,...,4}{
     \node[draw, fill=cyan!50, circle] at (0.5*\i, 0.5)  (a\i)  {};
   }
 
   \node[draw, fill=green!50, circle] at (0.5*0, -0.5)  (b0)  {};
   \node[draw, fill=green!50, circle] at (0.5*1, -0.5)  (b1)  {};
   \foreach \i in {2,...,4}{
     \node[draw, fill=orange!50, circle] at (0.5*\i, -0.5)  (b\i)  {};
   }
 
   \begin{pgfonlayer}{background}
   \foreach \i in {0,...,4}{
     \pgfmathtruncatemacro{\il}{mod(\i+1,5)}
     \ifthenelse{\i=0}{
     \draw[ultra thick,draw=red!50] ($ (a\i.south) + (0,0.01) $) -- ($ (b\il.north) - (0,0.025) $);
     \draw[ultra thick,draw=red!50] ($ (a\i.south) + (0,0.01) $)-- ($ (b\i.north) - (0,0.025) $);
     }{
      \draw[ultra thick,] ($ (a\i.south) + (0,0.01) $)-- ($ (b\il.north) - (0,0.025) $);
      \draw[ultra thick] ($ (a\i.south) + (0,0.01) $) -- ($ (b\i.north) - (0,0.025) $);
     }
   }
   \end{pgfonlayer}
 \end{tikzpicture}}
  \caption{Illustration of two orbits that are non-homogeneously connected on the left in blue and purple. On the right, fixing a vertex of one orbit indicated in red immediately partitions the other orbit into two orbits: the neighbors of the red vertex in green, and the non-neighbors in orange.} \label{fig:orbithomogeneous}
\end{figure}
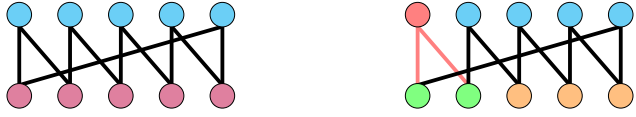
Since connected components can be computed in linear time in the size of a graph, and the size of the orbit graph is at most linear in the size of the original graph, we can therefore compute the finest direct disjoint decomposition in instance-linear time.
In a consecutive step, the generators could be split according to factors, producing a separable generating set, again in instance-linear time. This is done by separating each generator into the different factors.
 Finally,
 given the finest direct disjoint decomposition, we can again produce a joint graph/group pair for each factor, by outputting for a factor $H_i$ the induced subgraph $G'[H_i]$.  We summarize the above in a theorem:
\begin{theorem} Given a joint graph/group pair $(G, S)$ and orbit partition of $\langle S \rangle$, there is an instance-linear algorithm which computes the following:
  \begin{enumerate}
    \item The finest disjoint direct decomposition $\langle S \rangle = H_1 \times H_2 \times \dots\times H_m$.
    \item A separable generating set $S'$ with $\langle S' \rangle = G$. 
    \item For all factors $i \in \{1, \dots{}, m\}$ a joint graph/group pair $(G_i, S_i)$ with $\Aut(G_i) = \langle S_i \rangle = H_i$ in instance-linear time.
  \end{enumerate}
\end{theorem}
We recall that if the orbit partition of $\langle S \rangle$ is not yet available, we can compute it in instance-quasi-linear time.

\textbf{Domain Reduction to SAT Literals.}
For a SAT formula $F$, we can apply the above procedure to its model graph $\mathcal{M}(F)$.
However, as mentioned above, in SAT we are typically only interested in symmetries for a subset of vertices, namely the vertices that represent literals.
Therefore, we are specifically interested in the finest direct disjoint decomposition of the automorphism group reduced to literal vertices $\Aut(\mathcal{M}(F))|_{\Lit(F)}$.
The crucial point here is that when removing orbits that represent clauses, orbits of literal vertices can become independent and the disjoint direct decomposition can therefore become finer.
We cannot simply apply our algorithm for the induced group $\Aut(\mathcal{M}(F))|_{\Lit(F)}$ since this is not a joint graph/group pair. Of course we could apply the algorithm from \cite{DBLP:journals/jsc/ChangJ22} that computes finest disjoint direct decomposition for permutation groups in general.
However, we can detect \emph{some} forms of independence by simple means using the original joint graph/group pair. Indeed, we will describe an algorithm that checks in instance-linear time whether the parts in a given partition of the literals are independent. We can thus at least check whether a given partition induces a disjoint direct decomposition: 

\begin{restatable}{theorem}{satdecomposition}
Let~$F$ be a CNF-Formula and~$(\mathcal{M}(F),S)$ be a  joint graph/group pair for the model graph of~$F$. Given a partition~$\Lit(F) = L_1 \cup L_2 \cup\cdots \cup L_t$ of the literals of~$F$, the pair~$(\mathcal{M}(F),S)$, and its orbits, we can check in instance-linear time whether the partition induces a disjoint direct product (that is, whether~$\Aut(F) = \Aut(F)_{(L\setminus L_1)} \times \cdots \times \Aut(F)_{(L\setminus L_t)}$). 
\end{restatable}
\begin{proof}
  First, we check whether each~$L_i$ is a union of literal orbits. Otherwise, we do not have a disjoint direct product.
  
  We now argue that we can treat each clause orbit independently. Indeed, in the construction of the model graph clause vertices are never adjacent to other clause vertices. Moreover, literal vertices can only be adjacent to their negation or clauses. Thus $\Aut(F) = \Aut(F)_{(L\setminus L_1)} \times \cdots \times \Aut(F)_{(L\setminus L_t)}$ exactly if for every clause orbit $\Delta$ this decomposition is a disjoint direct product when we remove all clauses not in~$\Delta$.
  
  For each clause orbit $\Delta$ we will check this in time linear in $|C|\cdot |\Delta|$, where~$|\Delta|$ denotes the total number of clauses of the orbit and~$|C|$ is the number of literals in each clause. Thus, overall we will get an instance-linear time algorithm.
  
  We therefore assume from now on that $\Delta$ is the only clause orbit. We let $K = |\Delta|$ denote the total number of clauses of the orbit.
  We define for each part~$L_i$ in the given partition of literals the following:
  whenever there is a $C \in \Delta$ with $C \cap L_i = D$, we call $D$ a \emph{joint occurrence} of $L_i$.
  We denote by $n_i$ the number of joint occurrences of $L_i$.
  
  Note that all joint occurrences of~$L_i$ have the same size, since the symmetry group acts transitively on them.
  
  Also note that we can assume that literals appear together with their negation in the same part~$L_i$. Indeed, to obtain a disjoint direct product, they could only appear in different parts if they are fixed by all elements of~$\Aut(F)$, in which case we can put them into a new part containing only the two literals.
  
  We claim now that the partition~$L_1 \cup L_2 \cup\cdots \cup L_t$ induces a disjoint direct product (when it comes to $\Delta$) exactly if
  for all combined choices of occurrences~$D_i$ of~$L_i$ for each~$i\in \{1,\ldots,t\}$ there is a clause containing all~$D_i$ simultaneously. In other words, exactly all $N = \prod_{i=1}^t n_i$ combinations of joint occurrences are encoded in $\Delta$. 
  Note that this can be checked in instance-linear time by simply checking $K = N$.
  
  It remains to argue the claim. When $K = N$, every permutation in the projection~$\Aut(F)_{(L\setminus L_i)}$ extends in~$\Aut(\mathcal{M}(F))$ to~$L_i$ and to~$\Delta$. In fact, it is even possible to extend the maps by fixing all points of~$L_i$. Conversely, suppose~$D_1,\ldots,D_t$ are occurrences with~$D_i\in L_i$. 
  Choose some clause~$C$ in~$\Delta$. Let~$D_1', D'_2,\ldots,D'_t$ be the occurrences in~$C$, with~$D_i'\in L_i$. Since~$\Delta$ is an orbit, this means for each~$i$ some permutation~$\varphi_i \in  \Aut(\mathcal{M}(F))$ maps the set of literals~$D_i'$ to~$D_i$. If~$\Aut(F) = \Aut(F)_{(L\setminus L_1)} \times \cdots \times \Aut(F)_{(L\setminus L_t)}$ this implies that there is a permutation~$\varphi$ that simultaneously maps~$D_i'$ to~$D_i$ for all~$i\in \{1,\ldots,t\}$. Thus, some clause contains all~$D_i$ simultaneously.
  \end{proof}

\section{Natural Symmetric Action} \label{sec:snaction}
Before we can begin our discussion of the natural symmetric action, we need to discuss generating (nearly-)uniform random elements (see \cite{seress_2003}) of a given permutation group $\langle S \rangle$.
There is no known algorithm which produces uniform random elements of $\langle S \rangle$ in quasi linear time, even in computational group theory terminology \cite{seress_2003}.  
However, there are multiple ways to produce random elements,
most of which are proven to work well in practice, and can be implemented fairly easily \cite{seress_2003, GAP4}.
In this paper, we attempt to only make use of random elements sparingly.
Whenever we do, as is common in computational group theory, we do not consider the particular method used to generate them and simply denote the runtime of the generation with $\mu$.
Moreover, we discuss potential synergies in the SAT-symmetry context which might help to avoid random elements in practice, whenever applicable.

We now explain how to efficiently test whether a permutation group is a symmetric group in natural action.
Then, we describe more generally how to determine simultaneously for all orbits of a permutation group whether the induced action is symmetric in natural action.

Detecting symmetric permutation groups in their natural action is a well-researched problem in computational group theory. State-of-the-art practical implementations are available in modern computer algebra systems (such as in \cite{recog1.4.2}, or as described by \cite{unger2019fast}). 
Typically, a natural symmetric action is detected using a so-called probabilistic \emph{giant} test, followed by a test to ensure that the group is indeed symmetric. The tests work by computing (nearly) uniform random elements of the group and inspecting them for specific properties.

A permutation group is called a giant if it is the symmetric group or the alternating group in natural action. In many computational contexts, giants are by far the largest groups that appear, hence their name.
Because of this, giants form bottleneck cases for various algorithms and therefore often need to be treated separately.  To test whether a permutation group is a symmetric group in natural action, we first test whether the group is a giant. 
 
We leverage the following facts:
\begin{fact}[see~{\cite[Corollary 10.2.2.]{seress_2003}}]
  If a transitive permutation group of degree~$n$ contains an element with a cycle of length~$p$ for some prime~$p$ with~$n/2<p<n-2$ then~$G$ is a giant.
  \label{fact:gianttest}
  \end{fact}
  \begin{fact}[see~{\cite[Corollary 10.2.3.]{seress_2003}}]
    The proportion of elements in~$S_n$ containing a cycle of length~$p$ for a prime~$p$ with~$n/2<p<n-2$ is asymptotically~$1/\log(n)$.
    \label{fact:proportionofpcycle}
    \end{fact}
  Collectively, these statements show that we only need to generate few random elements of a group and inspect their cycle lengths to detect a giant. 
  To then distinguish between the alternating group and the symmetric group, we can check whether all generators belong to the alternating group.
  This can be attained using basic routines, such as examining the so-called parity of a generator (see \cite{seress_2003} for more details).
  Algorithm~\ref{alg:symact} generalizes the probabilistic test for a transitive group \cite{seress_2003} to a test which is performed simultaneously to check for a natural symmetric action on all the orbits of a group.
\newcommand{\algorithmicbreak}{\textbf{break}}

\begin{algorithm}
  \SetAlgoLined
	\SetAlgoNoEnd
	\caption{Compute whether orbits induce a natural symmetric action.} \label{alg:symact}
	\Fn{\SymmetricAction{$S, O$}}{
		\SetKwInOut{Input}{Input}
		\SetKwInOut{Output}{Output}
		\Input{generators $S$ with $\langle S \rangle = \Gamma \leq \Sym(\Omega)$, orbits $O \coloneqq \{\Delta_1, \dots{}, \Delta_m\}$ of $\Gamma$}
		\Output{set of orbits $S_O \subseteq O$, where $\Delta \in S_O$ induces a natural symmetric action}
    \tcp{first, we filter orbits which can at most be alternating}
    \For{$\Delta_i \in O$}{
      $t \coloneqq \top$\;
      \For{$s \in S$}{
        \lIf(\tcp*[f]{$\Delta_i$ cannot be alternating}){$s|_{\Delta_i} \not\in \Alt(\Delta_i)$}{
          $t \coloneqq \bot$
        }
      }
      \lIf(\tcp*[f]{$\Delta_i$ cannot induce symmetric action}){$t = \top$}{
        $O \coloneqq O \smallsetminus \{\Delta_i\}$
      }
    }
    \tcp{second, we test whether orbits are giants}
    $S_O \coloneqq \emptyset$\;
    \For(\tcp*[f]{repeat $c\log(|\Omega|)^2$ times}){$\_ \in \{1, \dots{}, c\log(|\Omega|)^2\}$}{
      $\varphi \coloneqq$ uniform random element of $\langle S|_{\cup_{\Delta_i \in S_O} \Delta_i} \rangle$\;
      \For{$\Delta_i \in O$}{
      $p \coloneqq$       cycle length of longest cycle in $\varphi$ on $\Delta_i$\;
      \If{$p > n/2 \wedge p < n-2 \wedge p \text{ prime}$}{
        $S_O \coloneqq  S_O \cup \{\Delta_i\}$\tcp*{$\Delta_i$ induces symmetric action}
        $O \coloneqq O \smallsetminus \{\Delta_i\}$\;
      }
      }
    }
    \Return{$S_O$}
	}
\end{algorithm}

\emph{Description of Algorithm~\ref{alg:symact}.} 
Overall, the algorithm samples uniform random elements of the group and checks whether the random elements exhibit long prime cycles (see Fact~\ref{fact:gianttest}). 
More precisely, the algorithm first distinguishes between potential alternating and symmetric groups on each orbit.
Then, it computes $d = c\log(|\Omega|)^2$ random elements.
For each random element and each orbit, we then apply the giant test (Fact~\ref{fact:gianttest}) to check whether the element certifies that the orbit induces a natural symmetric action.

\emph{Runtime of Algorithm~\ref{alg:symact}.} 
Let us assume access to random elements of the joint graph/group pair $(G, S)$ with $\langle S \rangle \leq \Sym(\Omega)$ in time $\mu$.
Assuming a random element can be produced in time $\mu$, the algorithm runs in worst-case time $\mathcal{O}(\log(|\Omega|)^2 (\mu + |\Omega|))$.

\emph{Correctness of Algorithm~\ref{alg:symact}.} Regarding the correctness of the algorithm, the interesting aspect is to discuss the error probability. We argue that the error probability is at most~$1/4$ if~$c$ is chosen to larger than~$2\ln(2)$. Practical implementations use $c = 20$ in similar contexts \cite{recog1.4.2}.

If an orbit~$\Delta$ does not induce a symmetric action, no error can be made. If an orbit~$\Delta$ induces a symmetric action, by Fact~\ref{fact:proportionofpcycle}, the probability that one iteration does not produce a long prime cycle for~$\Delta$ is at most~$(1-1/\log(n))$. Thus, the probability that none of the~$c\log(|\Omega|)^2$ iterations produces a long prime cycle for~$\Delta$ is bounded by~${(1-1/\log(n))}^{c\log(|\Omega|)^2}\leq (1/e)^{c\log(|\Omega|)}\leq  1/(4|\Omega|)$ since~$c>2\ln(2)$.
Since there can be at most $|\Omega|$ orbits, using the union bound, we get that the probability that the test fails for at least one of the orbits is at most~$|\Omega|\cdot 1/(4|\Omega|)= 1/4$.

When trying to compute a natural symmetric action on a graph/group pair, the following heuristics can be implemented in instance-linear time.

The first and most straightforward heuristic is that most of the time, it is fairly clear that the generators describe a natural symmetric action.
In particular, symmetry detection based on depth-first search seem to often produce generators that are transpositions. 
From these symmetric actions can be detected immediately.
This fact is implicitly used by column interchangeability heuristics in use today.
There are however many more ways to detect a natural symmetric action, many of which are implemented in modern computer algebra systems such as \cite{GAP4,recog1.4.2}.

Next, the symmetry detection preprocessor \sassy{} \cite{DBLP:journals/corr/abs-2302-06351} as well as the preprocessing used by \Traces{} sometimes detect a natural symmetric action on an orbit, by detecting certain structures of a graph. 
In these cases, the result should be immediately communicated to consecutive algorithms.
We can also use the graph structure to immediately discard orbits from the test of Algorithm~\ref{alg:symact}.
In particular, all orbits $\Delta$ where $G[\Delta]$ is neither the empty graph nor the complete graph cannot have a natural symmetric action.

Furthermore, the generators produced by \dejavu{} and \Traces{} are fairly random (for some parts even uniformly random \cite{DBLP:conf/esa/AndersS21}).
This means they should presumably work well with the probabilistic tests above.
Lastly, internally, symmetry detection tools often produce so-called Schreier-Sims tables~\cite{seress_2003}, which can be used to produce random elements effectively. 

Indeed, for our running example $F_E$, the natural symmetric action can be detected quite easily: let us consider the generators $S_E$ reduced to the orbit of $\{x_1, x_2, x_3\}$.
We observe that there is a generator $(x_1, x_2, x_3)$ and $(x_1, x_2)$.
While this is not a set of generators detected by current symmetry exploitation algorithms \cite{DBLP:conf/sat/Devriendt0BD16, DBLP:journals/mpc/PfetschR19}, this is indeed also an arguably obvious encoding of a natural symmetric action: 
for an orbit of size $n$, an $n$-cycle in conjunction with a transposition encodes a symmetric action. 

\section{Equivalent Orbits} \label{sec:rowinterchange}
Towards our overall goal to compute row interchangeability subgroups, we can now already determine which orbits induce a natural symmetric action.
By Lemma~\ref{lem:columninterchangeeq}, to detect elementary row interchangeability subgroups, we only miss a procedure for orbit equivalence.

We describe now how to compute equivalent orbits as the automorphism group of a special, purpose-built graph. 
Then, we give a faster algorithm computing equivalent orbits with natural symmetric action in instance-quasi-linear time, under mild assumptions. 
In particular, we can find all classes of equivalent orbits described by Lemma~\ref{lem:columninterchangeeq}.

\subsection{Cycle Type Graph}\label{subsec:cycle:type}
If two orbits are equivalent, they appear in every permutation in ``the same manner'': for example, if orbits~$\Delta_1$ and~$\Delta_2$ are equivalent then for every generator~$g$, the cycle types~$g$ induces on~$\Delta_1$ are the same as the cycle types~$g$ induces on~$\Delta_2$.
More generally, equivalent orbits must be equivalent with respect to \emph{every} generating set of the group. We introduce the \emph{cycle type graph} whose symmetries capture orbit equivalence.
This means we can use a symmetry detection tool to detect equivalent orbits.

For a group $\Gamma \leq \Sym(\Omega)$ and generating set $\langle S \rangle = \Gamma$, we define the \emph{cycle type graph} $\mathcal{C}(S)$ as follows.

Firstly, the \textbf{vertex set} of $\mathcal{C}(S)$ is the disjoint union~$V(\mathcal{C}(S))\coloneqq \Omega  \dot\cup \dot\bigcup_{g\in S}\supp (g)$. In other words, there is a vertex for each element of~$\Omega$ and there are separate elements for all the points moved by the generators. In particular if a point is moved by several generators there are several copies of the point. 

Secondly, the \textbf{edges} of $\mathcal{C}(S)$ are added as follows: each corresponding vertex for~$x\in \supp(g)$ is adjacent to the corresponding vertex for element~$x\in \Omega$ via an undirected edge. Furthermore, there are directed edges~$\{(x,g(x))\mid x\in \supp(g) \}$. In other words, for each generator $s_i \in S$, we add directed cycles for each cycle of the generator, as shown in Figure~\ref{fig:cyclegadget}.
In the following, we refer to directed cycles added in this manner as \emph{cycle gadgets}.

Thirdly, we define a \textbf{vertex coloring} for $\mathcal{C}(S)$. For this we enumerate the generators, i.e.,~$S = \{s_1, \dots{}, s_m\}$. We then color the vertices in~$\Omega$ with color~$0$ and an~$x\in \supp(g_i)$ is colored with~$(i,t)$, where~$t$ is the length of the cycle in~$g_i$ containing~$x$. With this, the cycle type graph is constructed in such a way that its automorphism structure captures equivalence of orbits, as is described in more detail below.

We record several observations on automorphisms of the cycle type graph.
\begin{restatable}{lemma}{cyclecentralizer}
  If $\Delta_1, \Delta_2$ are orbits of $\Gamma$ then there is some $b \in \Aut(\mathcal{C}(S))$ for which $b(\Delta_1) = \Delta_2$, if and only if $\Delta_1 \equiv \Delta_2$. \label{lem:cyclecentralizer}
\end{restatable}
\begin{proof}
  Let us first make some general remarks about elements $b \in \Aut(\mathcal{C}(S))$. Due to the nature of the graph, for every generator~$g$ we have~$b\circ g \circ b^{-1}=g$. In other words conjugation with~$b$ leaves the generating set invariant.
  Generally, conjugation is a group homomorphism: for all group elements~$g_1,g_2$ we have~$(b\circ g_1\circ b^{-1}) (b\circ g_2\circ b^{-1}) = b\circ (g_1 g_2)\circ b^{-1}$. Thus, by induction, since all elements $g \in \Gamma$ can be generated from the generators, we have $b\circ g \circ b^{-1}$.
  
  Now assume there is some $b \in \Aut(\mathcal{C}(S))$ for which $b(\Delta_1) = \Delta_2$.
  Then, by the previous discussion, for all elements of the group $\varphi \in \Gamma$ we have $b\circ \varphi\circ b^{-1} = \varphi$ and thus $b\circ \varphi = \varphi\circ b$.
  Hence, for any $\delta \in \Delta_1$, we have $\varphi(b(\delta)) = b(\varphi(\delta))$.
  Thus, the orbits are equivalent.

  On the other hand, assume two orbits $\Delta_1$ and $\Delta_2$ are equivalent, and let $b'$ denote the bijection $b' \colon \Delta_1 \to \Delta_2$. Define~$b$ to be the permutation for which~$b|_{\Delta_1}=b$,~$b|_{\Delta_2}= b'^{-1}$ and~$b|_{\Omega\setminus (\Delta_1\cup\Delta_2)} =\operatorname{id}$.
  Then~$b$ is an automorphism of $\mathcal{C}(S)$ with~$b(\Delta_1)= \Delta_2$.
\end{proof}
\noindent We may formulate the observations in group theoretic terms, giving the following lemma.
\begin{lemma}\label{lem:aut:of:cycle:type:graph:is:centralizer} Given a group $\Gamma \leq S(\Omega)$, its centralizer in the symmetric group $C_{\Sym(\Omega)}(\Gamma)$ and the cycle type graph $\mathcal{C}(\Gamma)$ are a joint graph/group pair, i.e., $\Aut(\mathcal{C}(\Gamma)) = C_{S(\Omega)}(\Gamma)$.
  \label{lem:graphgroupcentralizer}
\end{lemma}
Since the centralizer of~$\Sym(\Omega)$ in~$\Sym(\Omega)$ is trivial for~$|\Omega|> 2$, we get the following corollary.
\begin{corollary} If $\Gamma = \Sym(\Omega)$ with~$|\Omega|> 2$, the cycle type graph of $\Gamma$ is asymmetric. 
\end{corollary}
It follows from the corollary that for two equivalent orbits with a natural symmetric action the bijection~$b$ commuting with the generators and interchanging the orbits is in fact unique.

While the cycle type graph and the centralizer in the symmetric group
$(\mathcal{C}(\Gamma), C_{\Sym(\Omega)}(\Gamma))$
is a joint graph/group pair, we still have to compute the group: so far, we only have access to~$\mathcal{C}(\Gamma)$.
One option is a symmetry detection tool. However, this goes against our goal of invoking symmetry detection tools unnecessarily often --- and against our goal to find instance-linear algorithms.
Hence, instead of computing the entire automorphism group, our approach is to make due with less:
in the following, we enhance the cycle type graph in a way such that it becomes ``easy'' for color refinement. 
Color refinement is usually applied as a heuristic approximating the orbit partition of a graph.
However, on the enhanced graphs, we prove that color refinement is guaranteed to compute the orbit partition.
Then, we show that the orbit partition suffices to determine equivalent orbits.
Overall, these methods are only guaranteed to work for orbits with a natural symmetric action, as is the case in elementary row interchangeability groups. 

\subsection{Symmetries of Cycle Type Graph with Unique Cycles}
Our goal is now to enhance the cycle type graph such that color refinement is able to compute its orbit partition.
This in turn enables us to detect equivalent orbits, and in turn elementary row interchangeability groups.
Towards this goal, we first discuss an algorithm to compute unique cycles on orbits. 
Unique cycles are a key ingredient for our enhanced cycle type graph.
These cycles should be invariant with respect to an ordered generating set, a concept we explain first.

Given an ordered generating set~$(s_1,\ldots,s_m)$ for a permutation group $\Gamma$, i.e., $\Gamma = \langle \{s_i\mid i\in\{1,\ldots,m\}\}\rangle \leq \Sym(\Omega)$, a 
permutation~$\varphi\in \Sym(\Omega)$ 
\emph{fixes} the ordered generating set (point-wise under conjugation) if for all~$s_i$ we have~$\varphi \circ s_i \circ \varphi^{-1} = s_i$.

A permutation~$\pi\in \Sym(\Omega)$  is \emph{invariant} with respect to the ordered generating set if all permutations~$\varphi$ that fix~$(s_1,\ldots,s_m)$ under conjugation also fix~$(s_1,\ldots,s_m,\pi)$ under conjugation\footnote{In group theoretic terms,~$\pi$ is in $C_{\Sym(\Omega)}(C_{\Sym(\Omega)}(\langle s_1,\ldots,s_m\rangle))$.}. Note that all group elements in~$\pi\in \Gamma$ are invariant. However, there can be further invariant permutations. 

\begin{figure}
  \centering
  \begin{subfigure}{0.49\textwidth}
    \centering
    \scalebox{0.75}{\begin{tikzpicture}[scale=2,yscale=0.6]
   \clip(-0.125,0) rectangle (4.25,2);
    \foreach \i in {0,...,3}{
    \node[draw, fill=black, circle] at (0.5*\i, 0.5)  (a\i)  {};
    }
    \foreach \i in {0,...,3}{
    \node[draw, fill=black, circle] at (2.5+0.5*\i, 0.5)  (b\i)  {};
    }
    \foreach \i in {0,...,3}{
    \node[draw, fill=cyan!50, circle] at (0.5*\i, 1)  (aa\i)  {};
    \draw[ultra thick] (a\i) -- (aa\i);
    }
    \foreach \i in {0,...,3}{
    \node[draw, fill=cyan!50, circle] at (2.5+0.5*\i, 1)  (bb\i)  {};
    \draw[ultra thick] (b\i) -- (bb\i);
    }
  
    \foreach \i in {0,...,2}{
      \pgfmathtruncatemacro{\il}{mod(\i+1,4)}
      \draw[ultra thick,-stealth] (aa\i) -- (aa\il);
      \draw[ultra thick,-stealth] (bb\i) -- (bb\il);
    }
  
    \draw[ultra thick,-stealth] (aa3) to[out=90+45,in=90-45] (aa0);
    \draw[ultra thick,-stealth] (bb3) to[out=90+45,in=90-45] (bb0);
  \end{tikzpicture}}
    \caption{Cycle type graph.} \label{fig:cyclegadget}
  \end{subfigure}
\begin{subfigure}{0.49\textwidth}
  \centering
  \scalebox{0.75}{\begin{tikzpicture}[scale=2,yscale=0.6]

 \foreach \i in {0,...,7}{
  \node[draw, fill=orange!50, circle] at (0.25+0.5*\i, 0)  (c\i)  {};
  }

  \foreach \i in {0,...,3}{
  \node[draw, fill=black, circle] at (0.5*\i, 0.5)  (a\i)  {};
  }
  \foreach \i in {0,...,3}{
  \node[draw, fill=black, circle] at (2.5+0.5*\i, 0.5)  (b\i)  {};
  }
  \foreach \i in {0,...,3}{
  \node[draw, fill=cyan!50, circle] at (0.5*\i, 1)  (aa\i)  {};
  \draw[ultra thick] (a\i) -- (aa\i);
  }
  \foreach \i in {0,...,3}{
  \node[draw, fill=cyan!50, circle] at (2.5+0.5*\i, 1)  (bb\i)  {};
  \draw[ultra thick] (b\i) -- (bb\i);
  }

  \foreach \i in {0,...,2}{
    \pgfmathtruncatemacro{\il}{mod(\i+1,4)}
    \draw[ultra thick,-stealth] (aa\i) -- (aa\il) node [pos=0.35, inner sep=0, fill=white] {\small $1$};
    \draw[ultra thick,-stealth] (bb\i) -- (bb\il) node [pos=0.35, inner sep=0, fill=white] {\small $1$};
  }

  \foreach \i in {0,...,6}{
    \pgfmathtruncatemacro{\il}{mod(\i+1,8)}
    \draw[ultra thick,-stealth] (c\i) -- (c\il);
  }

  \foreach \i in {0,...,3}{
    \draw[ultra thick] (a\i) -- (c\i);
  }

  \foreach \i in {0,...,3}{
    \pgfmathtruncatemacro{\il}{\i+4}
    \draw[ultra thick] (b\i) -- (c\il);
  }

  \draw[ultra thick,-stealth] (c7) to[out=270-45,in=270+45,looseness=0.5] (c0);

  \draw[ultra thick,-stealth] (aa3) to[out=90+45,in=90-45] node [pos=0.5, inner sep=0, fill=white] {\small $5$} (aa0);
  \draw[ultra thick,-stealth] (bb3) to[out=90+45,in=90-45] node [pos=0.5, inner sep=0, fill=white] {\small $5$} (bb0);
\end{tikzpicture}}
  \caption{Enhanced cycle type graph} \label{fig:cyclegadgetcan}
\end{subfigure}
\caption{The cycle type graph and enhanced cycle type graph. The figure shows a cycle type gadget and canonical cyclic order for the permutation $(1 2 3 4)(5 6 7 8)$. Automorphisms of this graph are elements of the centralizer, which indicate equivalent orbits.} 
\end{figure}
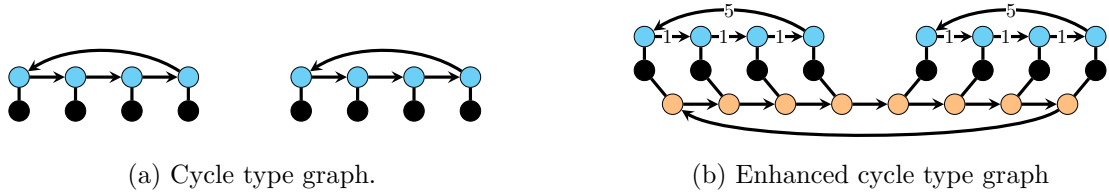

We say a permutation~$\pi$ has a \emph{unique cycle} if for some length~$\ell>1$, the permutation~$\pi$ contains exactly one cycle of~$\ell$. We now describe a two-step process. Step one is to compute an invariant permutation with a unique cycle. Step two is to use this to compute an invariant permutation with a cyclic order.

\textbf{Unique Cycle from Generators.} 
As a first step we now need an invariant \emph{unique cycle} to proceed. 
We argue how to compute such a cycle for orbits on which our group induces a natural symmetric action. 

We may use random elements to find a unique cycle. In fact, if we perform the giant test of Section~\ref{sec:snaction}, we get access to a unique cycle.
However, in that section we needed a prime length cycle. If we are only interested in unique cycles, not necessarily of prime length, this process terminates much more quickly:
\emph{Golomb's constant}~\cite{GOLOMB199898} measures, as $n \to \infty$, the probability that a random element $\varphi \in S_n$ has a cycle of length greater than $\frac{n}{2}$. The limit is greater than $\frac{4}{5}$.

In practice the existence of a unique cycle is a mild assumption: on the one hand some practical heuristics only apply if specific combinations of transpositions are present in the generators~\cite{DBLP:conf/sat/Devriendt0BD16, DBLP:journals/mpc/PfetschR19}. Each transposition is a unique cycle. 
On the other hand, randomly distributed automorphisms, as returned by \Traces{} and \dejavu{}, satisfy having a unique cycle with high probability, as argued above.

\textbf{Unique Cycle To Cyclic Order.}
We now assume we are given an invariant unique cycle $C$.
The idea is now to extend $C$ using the generators $S$ to a cycle which encompasses the entire orbit.
Crucially, the extension ensures that the result is still invariant (i.e. if we do this for all orbits simultaneously, the final permutation will be invariant).

We now describe the \emph{cycle overlap} algorithm, which gets as input a directed cycle $C$, as well as a collection of cycles $D_1, \dots{}, D_m$ which must be pair-wise disjoint. 
Furthermore, each $D_i$ must have one vertex in common with $C$.
The result is a cycle $C'$ that contains all vertices of all the cycles.
A formal description can be found in Algorithm~\ref{alg:overlap}.

\begin{figure}
  \centering
  \scalebox{0.8}{\begin{tikzpicture}[scale=2,xscale=0.95,yscale=0.8]

  \foreach \i in {0,...,3}{
  \node[draw, fill=cyan!50, circle] at (0.5*\i+0.75, 0.5)  (aa\i)  {\footnotesize $\i$};
  }

  \node[draw, fill=cyan!50, circle] at (0, 0)  (bb0)  {\footnotesize $0$};
  \node[draw, fill=orange!50, circle] at (0.5*1, 0)  (bb1)  {\footnotesize \Letter{1}};
  \node[draw, fill=cyan!50, circle] at (0.5*2, 0)  (bb2)  {\footnotesize $1$};
  \foreach \i in {3,...,4}{
    \pgfmathtruncatemacro{\il}{\i-1}
  \node[draw, fill=orange!50, circle] at (0.5*\i, 0)  (bb\i)  {\footnotesize \Letter{\il}};
  }

  \node[draw, fill=cyan!50, circle] at (0.5*5, 0)  (bb5)  {\footnotesize $2$};
  \node[draw, fill=orange!50, circle] at (0.5*6, 0)  (bb6)  {\footnotesize d};

  \foreach \i in {0,...,2}{
    \pgfmathtruncatemacro{\il}{mod(\i+1,4)}
    \draw[ultra thick,-stealth] (aa\i) -- (aa\il);
  }

  \foreach \i in {0,...,3}{
    \pgfmathtruncatemacro{\il}{mod(\i+1,5)}
    \draw[ultra thick,-stealth] (bb\i) -- (bb\il);
  }

  \draw[ultra thick,-stealth] (bb5) -- (bb6);

  \draw[ultra thick,-stealth] (aa3) to[out=90+60,in=90-60] (aa0);
  \draw[ultra thick,-stealth] (bb4) to[out=270-60,in=270+60] (bb0);
  \draw[ultra thick,-stealth] (bb6) to[out=270-60,in=270+60] (bb5);

  \draw[ultra thick, dashed, draw=black!50] (aa0) -- (bb0);
  \draw[ultra thick, dashed, draw=black!50] (aa1) -- (bb2);
  \draw[ultra thick, dashed, draw=black!50] (aa2) -- (bb5);

  \node[] at (5-1.5,0.25) (x)  {};
  \node[] at (5-0.75,0.25) (y)  {};
  \draw[-stealth, very thick] (x) to (y);

  \node[draw, fill=cyan!50, circle] at (-0.25+0.5*0   + 5, 0.25)  (cc0)  {\footnotesize $0$};
  \node[draw, fill=orange!50, circle] at (-0.25+0.5*1 + 5, 0.25)  (cc1)  {\footnotesize a};
  \node[draw, fill=cyan!50, circle] at (-0.25+0.5*2   + 5, 0.25)  (cc2)  {\footnotesize $1$};
  \node[draw, fill=orange!50, circle] at (-0.25+0.5*3 + 5, 0.25)  (cc3)  {\footnotesize b};
  \node[draw, fill=orange!50, circle] at (-0.25+0.5*4 + 5, 0.25)  (cc4)  {\footnotesize c};
  \node[draw, fill=cyan!50, circle] at (-0.25+0.5*5   + 5, 0.25)  (cc5)  {\footnotesize $2$};
  \node[draw, fill=orange!50, circle] at (-0.25+0.5*6 + 5, 0.25)  (cc6)  {\footnotesize d};
  \node[draw, fill=cyan!50, circle] at (-0.25+0.5*7   + 5, 0.25)  (cc7)  {\footnotesize $3$};

  \foreach \i in {0,...,6}{
    \pgfmathtruncatemacro{\il}{mod(\i+1,8)}
    \draw[ultra thick,-stealth] (cc\i) -- (cc\il);
  }
  \draw[ultra thick,-stealth] (cc7) to[out=270-60,in=270+60] (cc0);

\end{tikzpicture}}
  \caption{An illustration of the cycle overlap algorithm. Overlapping cycles $C = (0,1,2,3)$ and $D = \{(0,a,1,b,c), (2,d)\}$.} \label{fig:overlapcycle}
\end{figure}
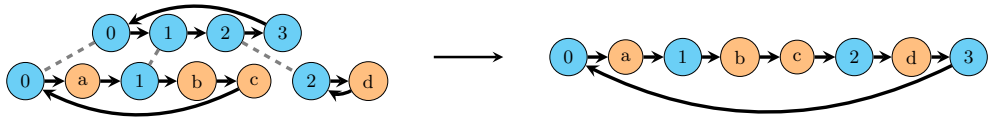

\emph{Description of Algorithm~\ref{alg:overlap}.} The algorithm first checks which vertices of $D= D_1\cup\cdots\cup D_m$ appear in $C$, and records them into the set $C'$.
Then, for each $c \in C'$ in the overlap of $C$ and $D$, the algorithm walks along the respective cycle~$D_i$ containing~$c$, and records all vertices it observes into $D'$.
It walks along the cycle until another $c' \in C'$ is reached (it may record the entire cycle~$D_i$, i.e., $c' = c$ may hold).
Finally, $D'$ is inserted as a path into $C$.

The output of the process is invariant under the cyclic orders involved. This means no matter in which order the cycles from~$D$ are processed, the algorithm always returns the same cyclic order.
Figure~\ref{fig:overlapcycle} illustrates the algorithm.

\emph{Runtime of Algorithm~\ref{alg:overlap}.} We may use a doubly-linked list structure for directed cycles $C$ and $D$, and an array $A$ to link vertices $V$ to their position in $C$ in time $\mathcal{O}(1)$. Assuming these data structures, inserting a $D'$ into $C$ can be performed in time $|D'|$. 
Indeed, with these data structures, we can implement the entire algorithm in time $\mathcal{O}(\Sigma_{i \in \{1, \dots{}, m\}} |D_i|)$.
We may also update the array $A$ to include the new vertices of $C$ added from $D$.

\SetKwFor{For}{for (}{)}{}
\begin{algorithm}[t]
  \SetAlgoLined
	\SetAlgoNoEnd
	\caption{The cycle overlap algorithm.} \label{alg:overlap}
	\Fn{\Overlap{C, D}}{
		\SetKwInOut{Input}{Input}
		\SetKwInOut{Output}{Output}
		\Input{directed cycle $C$, overlapping pair-wise disjoint directed cycles $D = \{D_1, \dots{}, D_m\}$ (with $\forall D_i \in D\colon D_i \cap C \neq \emptyset$, $\forall D_i, D_j \in D\colon i \neq j \implies D_i \cap D_j = \emptyset$)}
		\Output{directed cycle containing all vertices $C \cup D$}
		$C'\coloneqq \bigcup_{D_i \in D} C \cap D_i$\;
    \For{$c \in C'$}{
      $D'\coloneqq$ read $D$ from $c$ to next vertex of $C'$\;
      insert $D'$ after $c$ in $C$\;
    }
		\Return{$C$}\;
	}
\end{algorithm}

To get a unique cyclic order,
we repeatedly combine~$C$ with cycles appearing in generators that intersect~$C$. Every cycle in a generator only has to be processed once. Eventually~$C$ contains the entire orbit. With careful management of usage-lists of vertices in cycles of generators, the overall algorithm can be implemented in instance-linear time.

\subsection{Cyclic Order to Equivalent Orbits}
We finally describe how to find equivalent orbits, assuming invariant cyclic orders are given on the orbits.
An invariant cyclic order for the vertices of each orbit moves us one step closer to the orbits of the cycle type graph.
There are however still many potential bijections between orbits: indeed, we do not know how each cyclic order should be rotated.
We therefore describe a procedure to refine the cyclic order further.

We introduce the \emph{enhanced cycle type graph}~$\mathcal{C}'(S)$. 
We are provided an invariant cyclic order for each orbit $\Delta$ of $\Gamma$, which we denote by $C_\Delta$.
First, we add to the cycle type graph (Subsection~\ref{subsec:cycle:type}) a cycle gadget for each $C_\Delta$.
As before, we color the cycle gadget $C_\Delta$ according to its cycle length. 
Next, we enhance all other cycle gadgets using distance information of $C_\Delta$:
in every cycle gadget we mark each directed edge $v_1 \to v_2$ with the length of the (directed) path from~$v_1$ to~$v_2$ in $C_\Delta$ (see Figure~\ref{fig:cyclegadgetcan}).
We write $v_1 \xrightarrow{c} v_2$ whenever the path from $v_1$ to $v_2$ in $C_\Delta$ has length~$c$.
Note that while we use edge-labels for clarity, these can be encoded back into vertex colors (see \cite[Proof of Lemma~15]{DBLP:journals/jacm/KieferPS19}).

Just like with the cycle type graph, the automorphism group of the enhanced cycle type graph~$\mathcal{C}'(S)$ is the centralizer of~$\Gamma$ and~$(C_{\Sym(\Omega)}(\Gamma),\mathcal{C}'(\Gamma))$ is a joint graph/group pair (see Lemma~\ref{lem:aut:of:cycle:type:graph:is:centralizer}). 
However, it is easier to compute the orbit partition of~$\mathcal{C}'(S)$.

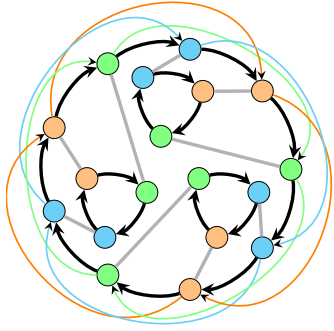
\begin{figure}
  \centering
  \scalebox{0.8}{\begin{tikzpicture}[scale=1.5]
    \clip(-1.8,-1.725) rectangle (2,2);
 \node [circle, draw=white, minimum size=4cm, ultra thick] (c) {};
  \foreach \i in {0,...,8}{
    \pgfmathtruncatemacro{\il}{mod(\i,3)}
    \ifthenelse{\il=1}{
      \node[draw, fill=orange!50, circle] at (c.\i*360/9) (a\i)  {};
    }{
      \ifthenelse{\il=2}{
      \node[draw, fill=cyan!50, circle] at (c.\i*360/9) (a\i)  {};
      }{
        \node[draw, fill=green!50, circle] at (c.\i*360/9) (a\i)  {};
      }
    }
  }

  \foreach \i in {0,...,8}{
    \pgfmathtruncatemacro{\il}{mod(\i+1,9)}
    \draw[ultra thick,-stealth,draw=black] (a\il) to[bend left,looseness=0.75] (a\i); 
  }

  \node [circle, draw=white, minimum size=1.125cm, ultra thick] at (0,0.75) (c0) {};
  \foreach \i in {0,...,2}{
    \ifthenelse{\i=1}{
      \node[draw, fill=orange!50, circle] at (c0.\i*360/3+240+17) (b\i)  {};
    }{
      \ifthenelse{\i=2}{
      \node[draw, fill=cyan!50, circle] at (c0.\i*360/3+240+17) (b\i)  {};
      }{
        \node[draw, fill=green!50, circle] at (c0.\i*360/3+240+17) (b\i)  {};
      }
    }
  }
  \foreach \i in {0,...,2}{
    \pgfmathtruncatemacro{\il}{mod(\i+1,3)}
    \draw[ultra thick,-stealth,draw=black] (b\il) to[bend left,looseness=0.75] (b\i); 
  }

  \node [circle, draw=white, minimum size=1.125cm, ultra thick] at (-0.62,-0.37) (c1) {};
  \foreach \i in {3,...,5}{
    \pgfmathtruncatemacro{\il}{mod(\i,3)}
    \ifthenelse{\il=1}{
      \node[draw, fill=orange!50, circle] at (c1.\i*360/3+17) (b\i)  {};
    }{
      \ifthenelse{\il=2}{
      \node[draw, fill=cyan!50, circle] at (c1.\i*360/3+17) (b\i)  {};
      }{
        \node[draw, fill=green!50, circle] at (c1.\i*360/3+17) (b\i)  {};
      }
    }
  }
  \foreach \i in {3,...,5}{
    \pgfmathtruncatemacro{\il}{mod(\i+1,3)+3}
    \draw[ultra thick,-stealth,draw=black] (b\il) to[bend left,looseness=0.75] (b\i); 
  }

  \node [circle, draw=white, minimum size=1.125cm, ultra thick] at (0.62,-0.37) (c2) {};
  \foreach \i in {6,...,8}{
    \pgfmathtruncatemacro{\il}{mod(\i,3)}
    \ifthenelse{\il=1}{
      \node[draw, fill=orange!50, circle] at (c2.\i*360/3+120+17) (b\i)  {};
    }{
      \ifthenelse{\il=2}{
      \node[draw, fill=cyan!50, circle] at (c2.\i*360/3+120+17) (b\i)  {};
      }{
        \node[draw, fill=green!50, circle] at (c2.\i*360/3+120+17) (b\i)  {};
      }
    }
  }
  \foreach \i in {6,...,8}{
    \pgfmathtruncatemacro{\il}{mod(\i+1,3)+6}
    \draw[ultra thick,-stealth,draw=black] (b\il) to[bend left,looseness=0.75] (b\i); 
  }

  \foreach \i in {0,...,8}{
    \draw[ultra thick,draw=black!30] (a\i) -- (b\i);
  }

  \foreach \i in {0,...,8}{
    \pgfmathtruncatemacro{\il}{mod(\i+3,9)}
    \pgfmathtruncatemacro{\ik}{mod(\i,3)}
    \ifthenelse{\ik=1}{
    \draw[thick,-stealth,draw=orange] (a\il) to[out=60+\i*360/9,in=53+\i*360/9,looseness=1.5] (a\i);
    }{
      \ifthenelse{\ik=2}{
    \draw[thick,-stealth,draw=cyan!50] (a\il) to[out=60+\i*360/9,in=53+\i*360/9,looseness=1.33] (a\i);
    }{
      \draw[thick,-stealth,draw=green!50] (a\il) to[out=60+\i*360/9,in=53+\i*360/9,looseness=1.125] (a\i);
    }
    }
  }
\end{tikzpicture}}
  \caption{Illustration of the automorphism used in Lemma~\ref{lem:cyclegraphcolorref}. The figure shows a canonical cycle $(1,2,3,4,5,6,7,8,9)$ with a generator $(1,2,3)(4,5,6)(7,8,9)$ in the enhanced cycle type graph (edge labels omitted). Arrows indicate the automorphism.} \label{fig:cyclegraphauto}
\end{figure}

In fact, our method of obtaining the orbits of $\mathcal{C}'(S)$ is rather straightforward: we apply the color refinement procedure to the enhanced cycle type graph $\mathcal{C}'(S)$. 
\begin{restatable}{lemma}{cyclegraphcolorref} Color refinement computes the orbit partition of the enhanced cycle type graph.\label{lem:cyclegraphcolorref}
\end{restatable}
\begin{proof}
  Let us make some observations on any equitable coloring $\pi$ of the enhanced cycle type graph.
  Consider the canonical cycle $C_\Delta$ of an orbit.
  For a directed cycle of size $n$ to be color-stable, it must be partitioned into $m$ equally sized colors of, say, size $c_s$, where $mc_s=n$.
  These must always appear in-order in the cyclic order.
  This is also true for every cycle gadget, in every generator.
  
  Let us now assume we have two vertices $v_1, v_2$ of a single orbit $\Delta$ with equal color.
  We can conclude the following:
  \begin{enumerate}
    \item In every generator $g$, either both $v_1 \in \supp(g)$ and $v_2 \in \supp(g)$ hold, or neither.
    \item In every generator $g$, $v_1$ and $v_2$ are in a cycle of equal size.
    \item In every generator $g$ where $v_1 \xrightarrow{c_1} v_1'$ and $v_2 \xrightarrow{c_2} v_2'$, $\pi(v_1') = \pi(v_2')$. Furthermore, their distance in the canonical cycle must be equal, i.e., $c_1 = c_2$.
    \item In every generator $g$, if $v_1$ and $v_2$ are contained in cycles $c_1$ and $c_2$, the vector of colors and distances starting from $v_1$ in $c_1$, and $v_2$ in $c_2$ must be identical. In other words, the previous property must hold transitively. 
  \end{enumerate}
  We call a sequence of vertices $v_1, v_2, \dots{}, v_{x}$ \emph{evenly spaced} in a cycle gadget $d$, whenever for each pair $v_i, v_{i+1\pmod x}$ with $i \in \{1, \dots{}, x\}$ the sum of distances in $d$ when going from $v_i$ to $v_{i+1\pmod x}$, i.e., $w = \Sigma_{j \in \{1,\dots{},y\}}w_j$ where $v_i \xrightarrow{w_1} v' \xrightarrow{w_2} \dots{} \xrightarrow{w_y} v_{i+1\pmod x}$, is the same.
  We call $w$ the \emph{spacing}.
  A sequence of vertices is evenly spaced in the canonical cycle, whenever the above holds assuming edge weights of $1$.
  
  Let $p$ be the permutation which rotates any canonical cycle by $m$ to the right, such that equally colored vertices are mapped onto each other.
  Intuitively, this means the canonical cycle is rotated by the least possible amount (see Figure~\ref{fig:cyclegraphauto}).
  We extend $p$ such that all vertices which correspond to a vertex of the canonical cycle are mapped accordingly.
  We show that $p$ must be an automorphism of the enhanced cycle type graph.
  
  By definition, the canonical cycle is mapped back to itself.
  It remains to be shown that every generator $g$, as well as its connections to the canonical cycle, are mapped back to themselves.
  
  Let $g \in S$ be a generator.  
  Let $D_1\dots{}D_m$ denote the cycle gadgets of $g$.
  First of all, note that if a color of vertices is represented in $g$, all of its respective vertices must be in $g$ (1).
  Furthermore, all cycles in which vertices of a color are present, must have the same size (2), and contain the same number of vertices of a color, in the same order (3) and (4).
  
  Indeed, vertices of a color $c$ must be evenly spaced in each cycle gadget $d_i$ due to (4). 
  This immediately implies that they are also evenly spaced in the canonical cycle.
  In fact, vertices of color $c$ in other cycle gadgets must also be evenly spaced with the same spacing as the first cycle, again due to (4).
  
  We now read vertices of $c$ in the cyclic order of the canonical cycle, say, $v_1, \dots{}, v_x$. 
  Using this, we also get a cyclic order of the cycle gadgets, in which cycle gadgets may appear multiple times.
  We denote this with $(d_1, d_2, \dots{}, d_x)$.
  We argue that if cycle gadgets do repeat, they must always repeat in the same order:
  however, since each cycle gadget contains the same number of vertices of $c$, and all of them are evenly spaced in the canonical cycle, this immediately follows.
  For example, vertices $v_1, v_2, v_3, v_4, v_5, v_6$ may lead to $(d_1, d_2, d_3, d_1, d_2, d_3)$, but not $(d_1, d_2, d_3, d_1, d_3, d_2)$. 
  Indeed, such an ordering would contradict the even spacing with respect to the canonical cycle. In the example, $v_1 \to v_4$ in $d_1$ and $v_3 \to v_5$ in $d_3$ would not be equidistant, hence, these vertices could be distinguished. 
  Naturally, the ordering $v_1, \dots{}, v_x$ must also respect the ordering of each cycle gadget individually, since in each cycle gadget vertices of a color are evenly spaced.
  
  Therefore, looking at color $c$, since $p$ maps the canonical cycle ``one to the right'', it maps the vertices of cycle $D_i$ to cycle $D_{i+1\pmod m}$.
  Moreover, as mentioned above, when looking at vertices of $c$, we know it also respects the order of each cycle gadget.
  
  Let us now consider the next color $c' = v_1', \dots{}, v_x'$ according to the cycle gadgets of $g$, i.e., $v_i \xrightarrow{c_i} v_i'$ for $i \in \{1, \dots{}, x\}$ in $g$.
  Immediately, we get that $c_i = c_j$ for all, $i, j \in \{1, \dots{}, x\}$.
  This means that $v_1', \dots{}, v_x'$ is also ordered correctly according to the canonical cycle.
  Indeed, if $v_i \to v_i'$ then $p(v_i) \to p(v_i')$ immediately follows. 
  The automorphism therefore maps the cycle gadgets $(d_1, d_2, \dots{}, c_x)$ back to themselves, and therefore also the generator $g$ back to itself.
  Hence, the automorphism $g$ maps all gadgets related to some orbit $\Delta$ back to itself.
  
  Let us now consider the case where vertices of two different orbits $v \in \Delta$ and $v' \in \Delta'$ are equally colored, i.e., $\pi(v) = \pi(v')$.
  Note that this can only be the case when $\Delta$ and $\Delta'$ are indeed equally sized.
  Indeed, then the arguments above hold true as well, and there must be an automorphism interchanging $\Delta$ and $\Delta'$:
  in particular, the automorphism maps $v$ to $v'$ (and vice versa), precisely mapping the canonical cycles onto each other $C_\Delta$ to $C_\Delta'$, starting at $v$ and $v'$ respectively.
  \end{proof}

\SetKwFor{For}{for (}{)}{}
\begin{algorithm}[t]
  \SetAlgoLined
	\SetAlgoNoEnd
	\caption{High-level procedure to obtain equivalent orbits.} \label{alg:eqorbits}
	\Fn{\Equivalent{$S, \Delta_1, \dots{}, \Delta_m, C_{\Delta_1}, \dots{}, C_{\Delta_m}$}}{
		\SetKwInOut{Input}{Input}
		\SetKwInOut{Output}{Output}
		\Input{group $\langle S \rangle = \Gamma \leq S(\Omega)$, orbit partition $\Delta_1, \dots{}, \Delta_m$ of $\Gamma$, unique cycles $C_{\Delta_1}, \dots{}, C_{\Delta_m}$}
		\Output{partition of $v \in \Omega$ into equivalent orbits}
		overlap each unique cycle $C_{\Delta_i}$ with $S$ to obtain canonical cycle $C_{\Delta_i}'$\;
    construct enhanced cycle type graph $\mathcal{C}'(S)$\;
    $\pi \coloneqq$ apply color refinement to $\mathcal{C}'(S)$\;
    \Return{$\pi$} 
	}
\end{algorithm}

Given our high-level procedure in Algorithm~\ref{alg:eqorbits}, and given that color refinement can be computed in quasi-linear time as previously discussed, leads to the following theorem:
\begin{theorem} Given access to a unique cycle per orbit, there is an instance-quasi-linear algorithm which computes for a joint graph/group pair $(G, S)$ a partition $\pi$ of equivalent orbits. 
Given two equivalent orbits $\Delta_1 \equiv \Delta_2$, there is an algorithm which computes from $\pi$ a corresponding matching $b \colon \Delta_1 \to \Delta_2$ such that for all $\varphi \in \Gamma$ and $\delta \in \Delta_1$, $\varphi(b(\delta)) = b(\varphi(\delta))$ in time $\mathcal{O}(|\Delta_1|)$. \label{lem:finallemcolorref}
\end{theorem}

\section{Conclusion and Future Work}
Exploiting our concept of joint graph/group pairs, we proposed new, asymptotically faster algorithms for the SAT-symmetry interface. 
However, most of the new concepts and approaches of this paper do not only apply to the domain of SAT, but also for example to MIP \cite{DBLP:journals/mpc/PfetschR19} and CSP~\cite{DBLP:conf/cp/FlenerFHKMPW02}.
More computational tasks should be considered in this context, the most prominent one arguably being pointwise stabilizers \cite{seress_2003}.

Our new algorithms exploit subroutines with highly efficient implementations available, but otherwise do not use any complicated data structures. We intend to implement the algorithms and integrate them into the symmetry detection preprocessor \textsc{sassy} \cite{DBLP:journals/corr/abs-2302-06351}.

Finally, in some classes of SAT instances, 
more complex symmetry structures may arise. Analyzing and taking advantage of these structures is potential future work.
For example, in the pigeonhole principle, \textsc{BreakID} finds overlapping row interchangeability groups and breaks these groups partially.
By virtue of being overlapping, the symmetry breaking constraints produced are \emph{not} guaranteed to be complete.

Another example for a complex symmetry structure is the \emph{wreath product} of two symmetric groups, i.e., $S_n\;wr\;S_m$. 
These wreath products naturally occur as the automorphism groups of tree-like structures. 
Procedures to detect and exploit such groups (e.g., by first using blocks of imprimitivity \cite{seress_2003} followed by the algorithms of this paper) could be of practical interest.
\bibliography{main}

\begin{thebibliography}{10}

\bibitem{DBLP:conf/dac/AloulMS03}
Fadi~A. Aloul, Igor~L. Markov, and Karem~A. Sakallah.
\newblock Shatter: efficient symmetry-breaking for boolean satisfiability.
\newblock In {\em Proceedings of the 40th Design Automation Conference, {DAC}},
  pages 836--839. {ACM}, 2003.

\bibitem{DBLP:conf/dac/AloulRMS02}
Fadi~A. Aloul, Arathi Ramani, Igor~L. Markov, and Karem~A. Sakallah.
\newblock Solving difficult {SAT} instances in the presence of symmetry.
\newblock In {\em Proceedings of the 39th Design Automation Conference, {DAC}},
  pages 731--736. {ACM}, 2002.

\bibitem{DBLP:conf/esa/AndersS21}
Markus Anders and Pascal Schweitzer.
\newblock Parallel computation of combinatorial symmetries.
\newblock In Petra Mutzel, Rasmus Pagh, and Grzegorz Herman, editors, {\em 29th
  Annual European Symposium on Algorithms, {ESA}}, volume 204 of {\em LIPIcs},
  pages 6:1--6:18. Schloss Dagstuhl - Leibniz-Zentrum f{\"{u}}r Informatik,
  2021.

\bibitem{DBLP:conf/icalp/AndersS21}
Markus Anders and Pascal Schweitzer.
\newblock Search problems in trees with symmetries: Near optimal traversal
  strategies for individualization-refinement algorithms.
\newblock In Nikhil Bansal, Emanuela Merelli, and James Worrell, editors, {\em
  48th International Colloquium on Automata, Languages, and Programming,
  {ICALP}}, volume 198 of {\em LIPIcs}, pages 16:1--16:21. Schloss Dagstuhl -
  Leibniz-Zentrum f{\"{u}}r Informatik, 2021.

\bibitem{DBLP:journals/corr/abs-2302-06351}
Markus Anders, Pascal Schweitzer, and Julian Stie{\ss}.
\newblock Engineering a preprocessor for symmetry detection.
\newblock {\em CoRR}, abs/2302.06351, 2023.

\bibitem{DBLP:journals/siamcomp/BeameKPS02}
Paul Beame, Richard~M. Karp, Toniann Pitassi, and Michael~E. Saks.
\newblock The efficiency of resolution and davis--putnam procedures.
\newblock {\em {SIAM} J. Comput.}, 31(4):1048--1075, 2002.

\bibitem{DBLP:conf/aaai/0001GMN22}
Bart Bogaerts, Stephan Gocht, Ciaran McCreesh, and Jakob Nordstr{\"{o}}m.
\newblock Certified symmetry and dominance breaking for combinatorial
  optimisation.
\newblock In {\em Thirty-Sixth {AAAI} Conference on Artificial Intelligence,
  {AAAI}}, pages 3698--3707. {AAAI} Press, 2022.

\bibitem{DBLP:journals/jsc/ChangJ22}
Mun~See Chang and Christopher Jefferson.
\newblock Disjoint direct product decompositions of permutation groups.
\newblock {\em J. Symb. Comput.}, 108:1--16, 2022.

\bibitem{DBLP:conf/kr/CrawfordGLR96}
James~M. Crawford, Matthew~L. Ginsberg, Eugene~M. Luks, and Amitabha Roy.
\newblock Symmetry-breaking predicates for search problems.
\newblock In Luigia~Carlucci Aiello, Jon Doyle, and Stuart~C. Shapiro, editors,
  {\em Proceedings of the Fifth International Conference on Principles of
  Knowledge Representation and Reasoning (KR'96)}, pages 148--159. Morgan
  Kaufmann, 1996.

\bibitem{DBLP:conf/dac/DargaLSM04}
Paul~T. Darga, Mark~H. Liffiton, Karem~A. Sakallah, and Igor~L. Markov.
\newblock Exploiting structure in symmetry detection for {CNF}.
\newblock In Sharad Malik, Limor Fix, and Andrew~B. Kahng, editors, {\em
  Proceedings of the 41th Design Automation Conference, {DAC} 2004, San Diego,
  CA, USA, June 7-11, 2004}, pages 530--534. {ACM}, 2004.

\bibitem{DBLP:conf/sat/Devriendt0B17}
Jo~Devriendt, Bart Bogaerts, and Maurice Bruynooghe.
\newblock Symmetric explanation learning: Effective dynamic symmetry handling
  for {SAT}.
\newblock In Serge Gaspers and Toby Walsh, editors, {\em Theory and
  Applications of Satisfiability Testing - {SAT}}, volume 10491 of {\em LNCS},
  pages 83--100. Springer, 2017.

\bibitem{DBLP:conf/sat/Devriendt0BD16}
Jo~Devriendt, Bart Bogaerts, Maurice Bruynooghe, and Marc Denecker.
\newblock Improved static symmetry breaking for {SAT}.
\newblock In Nadia Creignou and Daniel~Le Berre, editors, {\em Theory and
  Applications of Satisfiability Testing - {SAT}}, volume 9710 of {\em LNCS},
  pages 104--122. Springer, 2016.

\bibitem{DBLP:conf/cp/FlenerFHKMPW02}
Pierre Flener, Alan~M. Frisch, Brahim Hnich, Zeynep Kiziltan, Ian Miguel,
  Justin Pearson, and Toby Walsh.
\newblock Breaking row and column symmetries in matrix models.
\newblock In Pascal~Van Hentenryck, editor, {\em Principles and Practice of
  Constraint Programming - {CP}}, volume 2470 of {\em LNCS}, pages 462--476.
  Springer, 2002.

\bibitem{GAP4}
The GAP~Group.
\newblock {\em {GAP -- Groups, Algorithms, and Programming, Version 4.12.2}},
  2022.

\bibitem{GOLOMB199898}
Solomon~W. Golomb and Peter Gaal.
\newblock On the number of permutations on n objects with greatest cycle length
  k.
\newblock {\em Advances in Applied Mathematics}, 20(1):98--107, 1998.

\bibitem{Grayland2009MinimalOC}
Andrew Grayland, Christopher Jefferson, Ian Miguel, and Colva~M. Roney-Dougal.
\newblock Minimal ordering constraints for some families of variable
  symmetries.
\newblock {\em Annals of Mathematics and Artificial Intelligence}, 57:75--102,
  2009.

\bibitem{DBLP:conf/tapas/JunttilaK11}
Tommi~A. Junttila and Petteri Kaski.
\newblock Conflict propagation and component recursion for canonical labeling.
\newblock In Alberto Marchetti{-}Spaccamela and Michael Segal, editors, {\em
  Theory and Practice of Algorithms in (Computer) Systems - First International
  {ICST} Conference, {TAPAS}}, volume 6595 of {\em LNCS}, pages 151--162.
  Springer, 2011.

\bibitem{DBLP:journals/jacm/KieferPS19}
Sandra Kiefer, Ilia Ponomarenko, and Pascal Schweitzer.
\newblock The weisfeiler-leman dimension of planar graphs is at most 3.
\newblock {\em J. {ACM}}, 66(6):44:1--44:31, 2019.

\bibitem{DBLP:conf/mfcs/KieferSS15}
Sandra Kiefer, Pascal Schweitzer, and Erkal Selman.
\newblock Graphs identified by logics with counting.
\newblock In Giuseppe~F. Italiano, Giovanni Pighizzini, and Donald Sannella,
  editors, {\em Mathematical Foundations of Computer Science 2015 - 40th
  International Symposium, {MFCS}, Part {I}}, volume 9234 of {\em LNCS}, pages
  319--330. Springer, 2015.

\bibitem{DBLP:conf/cp/KirchwegerS21}
Markus Kirchweger and Stefan Szeider.
\newblock {SAT} modulo symmetries for graph generation.
\newblock In Laurent~D. Michel, editor, {\em 27th International Conference on
  Principles and Practice of Constraint Programming, {CP}}, volume 210 of {\em
  LIPIcs}, pages 34:1--34:16. Schloss Dagstuhl - Leibniz-Zentrum f{\"{u}}r
  Informatik, 2021.

\bibitem{McKay81practicalgraph}
Brendan~D. McKay.
\newblock Practical graph isomorphism.
\newblock In {\em 10th. Manitoba Conference on Numerical Mathematics and
  Computing (Winnipeg, 1980)}, pages 45--87, 1981.

\bibitem{DBLP:journals/jsc/McKayP14}
Brendan~D. McKay and Adolfo Piperno.
\newblock Practical graph isomorphism, {II}.
\newblock {\em J. Symb. Comput.}, 60:94--112, 2014.

\bibitem{recog1.4.2}
M.~Neunh{\"o}ffer, {\a'A}.~Seress, and M.~Horn.
\newblock {recog}, a package for constructive recognition of permutation and
  matrix groups, {V}ersion 1.4.2.
\newblock \href {https://gap-packages.github.io/recog}
  {\texttt{https://gap-packages.github.io/}\discretionary
  {}{}{}\texttt{recog}}, Sep 2022.
\newblock GAP package.

\bibitem{DBLP:journals/mpc/PfetschR19}
Marc~E. Pfetsch and Thomas Rehn.
\newblock A computational comparison of symmetry handling methods for mixed
  integer programs.
\newblock {\em Math. Program. Comput.}, 11(1):37--93, 2019.

\bibitem{DBLP:conf/wea/Piperno18}
Adolfo Piperno.
\newblock Isomorphism test for digraphs with weighted edges.
\newblock In Gianlorenzo D'Angelo, editor, {\em 17th International Symposium on
  Experimental Algorithms, {SEA}}, volume 103 of {\em LIPIcs}, pages
  30:1--30:13. Schloss Dagstuhl - Leibniz-Zentrum f{\"{u}}r Informatik, 2018.

\bibitem{DBLP:conf/cp/Puget03}
Jean{-}Francois Puget.
\newblock Symmetry breaking using stabilizers.
\newblock In Francesca Rossi, editor, {\em Principles and Practice of
  Constraint Programming - {CP}}, volume 2833 of {\em LNCS}, pages 585--599.
  Springer, 2003.

\bibitem{DBLP:journals/constraints/Sabharwal09}
Ashish Sabharwal.
\newblock Symchaff: exploiting symmetry in a structure-aware satisfiability
  solver.
\newblock {\em Constraints An Int. J.}, 14(4):478--505, 2009.

\bibitem{DBLP:series/faia/Sakallah21}
Karem~A. Sakallah.
\newblock Symmetry and satisfiability.
\newblock In Armin Biere, Marijn Heule, Hans van Maaren, and Toby Walsh,
  editors, {\em Handbook of Satisfiability - Second Edition}, volume 336 of
  {\em Frontiers in Artificial Intelligence and Applications}, pages 509--570.
  {IOS} Press, 2021.

\bibitem{seress_2003}
Ákos Seress.
\newblock {\em Permutation Group Algorithms}.
\newblock Cambridge Tracts in Mathematics. Cambridge University Press, 2003.

\bibitem{unger2019fast}
William~R. Unger.
\newblock Fast detection of giant permutation groups.
\newblock {\em CoRR}, abs/1905.09431, 2019.
\newblock arXiv.

\end{thebibliography}
\bibliographystyle{plain}

\end{document}